\documentclass[a4paper]{article}
\usepackage[utf8]{inputenc}
\usepackage[T1]{fontenc}
\usepackage{amsmath}
\usepackage{amsfonts}
\usepackage{amssymb, amsthm}
\usepackage{graphicx, subfig}
\usepackage{tikz}
\usepackage{pgfplots}
\usepackage[noend,ruled,vlined, linesnumbered]{algorithm2e}
\usepackage{filecontents}

\newcommand{\cC}{\mathcal{C}}
\newcommand{\cF}{\ensuremath{\mathcal{F}}}
\newcommand{\cN}{\ensuremath{\mathcal{N}}}
\newcommand{\cS}{\ensuremath{\mathcal{S}}}

\newcommand{\cU}{\ensuremath{\mathcal{U}}}
\newcommand{\bt}{{\bf t}}
\newcommand{\bF}{{\bf F}}

\newcommand{\cL}{\mathcal{L}}

\newcommand{\bs}{\mathbf{s}} 
\newcommand{\bx}{\mathbf{x}} 
\newcommand{\bd}{\mathbf{d}} 
\newcommand{\bA}{\mathbf{A}} 
\newcommand{\bD}{\mathbf{D}}

\newcommand{\R}{\ensuremath{\mathbb{R}}}
\newcommand{\N}{\ensuremath{\mathbb{N}}}
\newcommand{\A}{\ensuremath{\mathbb{A}}}

\newcommand{\conv}{*}
\newcommand{\deconv}{\oslash}

\newcommand{\pr}{\mathrm{prec}}
\newcommand{\su}{\mathrm{succ}}

\newcommand{\sint}[1]{\mathbb{N}_{#1}}
\newcommand{\fl}{\mathrm{Fl}}

\usetikzlibrary{shapes}

\newtheorem{theorem}{Theorem}
\newtheorem{lemma}{Lemma}
\newtheorem{example}{Example}
\newtheorem{definition}{Definition}
\newtheorem{corollary}{Corollary}

\title{Trade-off between accuracy and tractability of network calculus in FIFO networks}
\author{Anne Bouillard \\
	{\tt{anne.bouillard@huawei.com}}\\
	Huawei Technologies France, \\20 quai du Point du Jour, \\92100 Boulogne-Billancourt}

\begin{document}

\maketitle

\begin{abstract}
	Computing accurate deterministic performance bounds is a strong need for communication technologies having strong requirements on latency and reliability. Beyond new scheduling protocols such as TSN, the FIFO policy remains at work within each class of communication. 
	
	In this paper, we focus on computing deterministic performance bounds in FIFO networks in the network calculus framework. We propose a new algorithm based on linear programming that presents a trade-off between accuracy and tractability. This algorithm is first presented for tree networks. In a second time, we generalize our approach and present a linear program for computing performance bounds for arbitrary topologies, including cyclic dependencies. Finally, we provide numerical results, both of toy examples and real topologies, to assess the interest of our approach.
\end{abstract}

\section{Introduction}

The aim of new communication technologies is to provide deterministic
services, with strong requirements on buffer occupancy, latency and reliability. An example of such a standard under discussion is Time-Sensitive Networking (TSN), which is part of the 802.1 working group~\cite{TSN} and has potential applications to industrial and automotive networks. In this new communication paradigm, critical traffic (having strong delay and reliability requirements) and best-effort traffic can share switches and routers. Even if scheduling policies have been defined to cope with these heterogeneous traffic classes, it is a necessity to develop tools for accurately dimensioning the bandwidth allocated to each class. 

\medskip

Properly dimensioning a network relies on the ability to compute accurate performance bounds (delay or buffer occupancy) in networks. As far as deterministic performance bounds are concerned, one popular theory is 
network calculus, which is based on the (min, plus) semi-ring. Elements of the network, such as the traffic flows and switches, are described by {\em curves},  and upper bounds of the performances  (delay, buffer occupancy) are computed from this description. This theory has already been successfully applied to various types of networks. One can cite switched network~\cite{Cruz1995}, Video-on-Demand~\cite{MR2006}, AFDX (Avionics Full Duplex) networks~\cite{BNOT2010}, TSN/AVB~\cite{MSMB18, ZPZ0}.

\medskip

Different solutions have recently been proposed to analyze these types of networks with network calculus. It is first required to give a precise modeling of the scheduling policy (priorities, processor sharing scheduling such as DRR (Deficit Round Robin) \cite{BSS12}, WRR (Weighted Round Robin), to deduce network guarantees for flows scheduled in the same {\em class}, where the FIFO (First In First Out) policy is at work. Being able to compute accurate performance bounds is FIFO networks is then crucial. 

Recent work focus on the analysis of FIFO networks, and their main goal is to reduce the computational cost for deriving performance guarantees (upper bounds of worst-case delay). For example, Mohammadpour et al. propose in~\cite{MSL19} propose a modeling of TNS, and the insertion of regulators~\cite{LeB18} to control the arrival processes at each router; Thomas et al. compare in~\cite{TBM19} the analysis with partial insertion regulators (from complete to none) using TFA++ (total flow analysis) proposed in~\cite{ML17}. These analyses have a very low complexity, which allow the analysis of large-scale networks, but can have pessimistic bounds. 

Other works focus on the accuracy of the bounds computed, in order to get the tightest result possible. From the first paper on network calculus, phenomena such as the {\em pay burst only once} and the {\em pay multiplexing only once} have been exhibited, and each time they led to improvements of the performance bounds. More recently, algorithms based on linear programming have been proposed in~\cite{BS12, BS15} to compute tight bounds in FIFO networks, but the complexity of these algorithms is too high to be used in most of the networks.

Nevertheless, some networks are not so large that they require very low complexity performance bounds. For example the linear programming method can be used to improve the computation of the performance bounds in smaller networks, such as industrial networks. For example, the TSN industrial network presented in~\cite{ZCWW19} has less than 20 nodes, where every flows cross at most 5 routers. The performances of these network could benefit from a more precise analysis at a small computational cost, even if this would be out of reach for larger networks. 

\paragraph{Objective and contributions.}
The objective of this paper is to explore a solution in-between these two extreme, that could be both tractable and lead to accurate bounds. 
We introduce a new polynomial size linear programming to compute performance bounds in FIFO networks, which could present a good trade-off between complexity and accuracy to analyze medium-size networks. 
Furthermore, we compare this algorithm with different network calculus methods. More precisely, our contributions are the following.

\begin{enumerate}
	\item We first propose a simplified model (regarding that of~\cite{BS12}) for a linear program computing bounds in FIFO trees. This model can also take into account the shaping of transmission links. While losing some accuracy, this algorithms is more tractable, and achieve better performances bounds than the other methods in the literature. 
	\item We generalize the linear programming algorithms to network with cyclic dependencies, improving the stability region of the other existing methods. 
	\item We compare our algorithms against the literature in both toy examples (tandems and rings) and real-world use-cases. 
\end{enumerate}

The rest of the paper is organized as follows. First, the network calculus framework and our network model are briefly recalled in Section~\ref{sec:framework}. The state of the art on network calculus for FIFO networks is described in Section~\ref{sec:soa}. In Section~\ref{sec:tree}, we present the first contribution of the paper, that is, a new linear programming proposition to compute performance bounds in FIFO tree networks, in polynomial time. This approach is generalized in Section~\ref{sec:ff} and~\ref{sec:cyclic} respectively to the case of feed-forward networks and networks with cyclic dependencies. Finally, we compare the new algorithm with the state of the art in several examples in Section~\ref{sec:numerical} before concluding. 

\section{Network calculus framework}
\label{sec:framework}
In this section, we recall the network calculus framework and present the basic results that will be used in the next parts of the paper. More details about the framework can be found in~\cite{BBL18, LT2001, Chang2000}. 

We will use the following notations: $\R_+$ is the set of non-negativea reals, for all $n\in \N$, $\N_n = \{1,\ldots n\}$,   and for all $x \in \R$, $(x)_+ = \max(0, x)$.
\subsection{Arrival and service curves}
\subsubsection{Data processes and arrival curves.}
Flows of data are represented by 
cumulative processes. More precisely, if $A$ represents a flow at a certain point
in the network, $A(t)$ is the amount of data of that flow crossing
that point during the time interval $[0,t)$, with the convention $A(0)=0$. The cumulative processes are non-decreasing, left-continuous and null at zero. We denote by $\cF$ the set of such functions. 

A flow $A$ is constrained by the arrival curve $\alpha$,
or is $\alpha$-constrained, if $$\forall s,t\in\R_+ \text{ with }s\leq t,
\quad A(t) - A(s) \leq \alpha(t-s).$$ 
In the following we will mainly consider  {\em token-bucket}
functions: $\gamma_{b,r}:0\mapsto 0;~t\mapsto b + r t$, if $t>0$.  The
{\em burst} $b$ can be interpreted as the maximal amount of data
that can arrive simultaneously and the {\em arrival rate}~$r$ as a maximal long-term
arrival rate of data.

\subsubsection{Servers and service curves.}

An $n$-server $\cS\subseteq \cF^n \times \cF^n$ (illustrated for $n=1$ in Figure~\ref{fig:server}) is a relation between $n$ arrival processes $(A_i)_{i=1}^n$ and $n$ departure processes $(D_i)_{i=1}^n$ such that $A_i\geq D_i$ for all $i\in\sint{n}$. The latter inequality models the causality of the system (no data is created inside the system).

\begin{figure}[htbp]
		\centering
		\begin{tikzpicture}[server/.style={shape=rectangle,draw,minimum height=.8cm,inner xsep=4ex}]
		\node[server,name=S1] at (0,0) {$\cS$};
		\draw[->, blue] (-1.2,-0) node[left] {$A$} -- (1.5,0) node[right] {$D$};
		\end{tikzpicture}
\caption{Server model.}
\label{fig:server}
\end{figure}
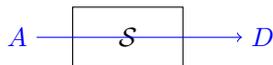

The role of a service curve is to constrain the relation between the
inputs of a server and its outputs. 

We say that~$\beta\in\cF$ is a {\em service curve} for 1-server~$\cS$ if
\begin{equation}
\label{eq:ssc}
\forall (A,D) \in \cS,~\ A \geq D \geq A\conv\beta,
\end{equation}
where $*$ is the (min,plus) convolution: for all $t\geq 0$, $A\conv \beta(t) = \inf_{0\leq s\leq t} A(s) + \beta(t-s)$.  
In the following we will use
\begin{itemize}
\item  the {\em rate-latency}
service curves: $\beta_{R,T}:t\mapsto R(t-T)_+$, where $T$ and $R$ can be roughly interpreted as $T$ is the
{\em latency} until the server becomes active and $R$ is its minimal
{\em service rate} after this latency;  
\item the {\em pure delay} service curve: $\delta_{d}: t\mapsto 0$ if $t\leq d$; $t\mapsto +\infty$ if $t>d$. We have $A\conv \delta_d(t) = A((t-d)_+)$
\end{itemize}

\paragraph{Service curve} An $n$-server $\cS$ offers a service curve $\beta$ if it offers the service curve $\beta$ for the aggregated flows: for all $((A_i), (D_i)) \in \cS$,  $(\sum_{i=1}^m A_i) \geq (\sum_{i=1}^m D_i) \conv \beta$.
We call the flow with arrival process $\sum_{i=1}^m A_i$ the aggregate process of flows $1,\ldots,n$. 

\paragraph{FIFO service curve} In this paper, we assume that the service policy in this system in FIFO (First-In-First-Out): data are served in their arrival order. It is possible to find service guarantees for individual flows.

\begin{theorem}[{\cite[Proposition 6.2.1]{LT2001}}]
	\label{th:fifo}
	Consider a FIFO server with service curve $\beta$, crossed
	by two flows with respective arrival curves $\alpha_1$ and $\alpha_2$. For all $\theta\geq 0$, $\beta_{\theta}$ is a residual service curve for the first flow, with 
	$$\beta_{\theta} = [\beta - \alpha_2\conv\delta_{\theta}]_+ \wedge \delta_{\theta}.$$ 		
\end{theorem}
	
One can notice that the service curves computed when $\theta$ is varying are not comparable, and lead to different performances. 

\paragraph{Greedy shapers} In most networks, the transmission rate is physically limited by the capacity of a wire or a channel, which limits the quantity of data that can be transmitted to the next server. This phenomenon is taken into account by {\em greedy shapers}.
Let $B$ be a cumulative process, crossing a leaky-bucket greedy shaper $\sigma: t\mapsto L + Ct$. The output process is $D = B\conv \sigma$. Here $C$ represents the maximum capacity of the server, and $L$ can represent a packet length, hence take into account the packetization effect. 

A server whose transmission rate is limited by a token-bucket greedy shaper can then be modeled by a system that is composed of a server $\beta$ and a greedy shaper $\sigma$, as depicted on Figure~\ref{fig:shap}. We will always assume that $\sigma \geq \beta$, which is not a restriction since the service offered to a flows is limited by the physical limitations of the server. 

Consider a system consisting in a 1-server with service curve $\beta$ followed by a greedy-shaper $\sigma$. The departure process then satisfies:
$$D = B\conv \sigma \geq (A \conv \beta) \conv \sigma = A \conv (\beta \conv \sigma) = A\conv \beta,$$ 
where the last equality comes from $\beta\leq \sigma$. 

\begin{figure}[htbp]
	\centering
	\begin{tikzpicture}[server/.style={shape=rectangle,draw,minimum height=.8cm,inner xsep=4ex}]
	\node[server,name=S1] at (0,0) {$\beta$};
	\node[draw, circle] (sh) at (2, 0) {$\sigma$};
	\draw[->, blue] (-1.2, 0) node[left] {$A$} -- (S1) -- node[above] {$B$} (sh) -- (3, 0)node[right] {$D$};
\end{tikzpicture}
\caption{Shaping of the output process.}
\label{fig:shap}
\end{figure}
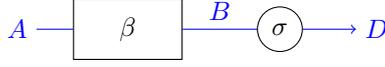

As a consequence, the whole system still offers $\beta$ as service curve.

\subsubsection{Output arrival curve}
Departure process are also characterized by an arrival curve. Such an arrival curve can be computed in function of the arrival curve of the arrival process aand the shaping and service curves of the server. 
\begin{theorem}[{\cite[Theorem 5.3]{BBL18}}]
	\label{th:output}
	Suppose that $A$ is $\alpha$-constrained and crosses a server offering the service curve $\beta$ and with greedy shaper $\sigma$. Then the departure process $D$ is $\alpha\deconv \beta \wedge \sigma$-constrained, where $\deconv$ is the (min, plus)-deconvolution: $\alpha\deconv\beta(t) = \sup_{u\geq 0} \alpha(t+u) - \beta(u)$.
\end{theorem}

In case of token-buckets arrival curve and greedy-shaper, the departure process then has two token-bucket constraints: $\alpha\deconv\beta$ and $\sigma$. 

In the case of a token-bucket arrival curve $\alpha = \gamma_{b,r}$ and rate-latency service curve $\beta = \beta_{R,T}$ with $R>r$, one has $\alpha\deconv\beta = \gamma_{b + tT,r}$.

\subsection{Performance guarantees in a server}

\paragraph{Backlog and delay} Let $\cS$ be a 1-server and $(A,D)\in\cS$.  The backlog of that server at time
$t$ is $b(t) = A(t) - D(t)$. The worst-case backlog is then
$b_{\max}=\sup_{t\geq 0} b(t)$.

We denote $b_{\max}(\alpha,\beta)$
the maximum backlog 
that can be obtained for an $\alpha$-constrained flow crossing
a server offering the service curve $\beta$. It can be shown to be the maximum vertical distance between $\alpha$ and $\beta$. For example, we have
$b_{\max}(\gamma_{b,r} ,\beta_{R,T}) =  b+rT $ if $r\leq R$.

The  delay of data exiting at time $t$ is $d(t) = \sup\{d\geq 0\mid A(t-d) - D(t)\}$. The worst-case delay is then
$d_{\max}=\sup_{t\geq 0} d(t)$.

We denote $d_{\max}(\alpha,\beta)$
the maximum delay 
that can be obtained for an $\alpha$-constrained flow crossing
a server offering the service curve $\beta$. It can be shown to be the maximum horizontal distance between $\alpha$ and $\beta$. For example, we have
$d_{\max}(\gamma_{b,r} ,\beta_{R,T}) =  T + \frac T R$ if $r <  R$.

Backlog and delay are illustrated on Figures~\ref{fig:proc} and~\ref{fig:perf}.

\begin{figure}[htbp]
	\centering
	\subfloat[\label{fig:proc}Processes]{
		\centering
		\begin{tikzpicture}
			\draw[->] (0,0) -- (3,0)node[below, pos = 0.9] {time};
			\draw[->] (0,0) -- (0,2) node[left, pos = 0.8] {\rotatebox{90}{data}};
			\draw[red] (0,0) -- (0, 0.5) -- (1, 1) -- (1.5,1.5) -- (3,1.7) node[above = -0.05cm, pos = 0.8] {$A$};
			\draw[blue] (0,0) -- (0.5,0) -- (0.8333, 0.667) -- (1.5, 1) --  (2,1.5) -- (3,1.63) node[below, pos = 0.5] {$D$};
			\draw[<->, yellow!70!black, thick]   (0.5, 0) node[below]{$t$} -- (0.5, 0.75) node[left, pos = 0.5] {\rotatebox{90}{\tiny $b(t)$}};
			\draw[green!70!black, <->, thick] (1.2,1.2) -- (1.7,1.2) node[above = -0.1cm, pos = 0.75] {\tiny $d(u)$};
			\draw[green!70!black, dotted] (1.7,1.2) -- (1.7, 0) node[below = 0.05cm] {$u$};
	\end{tikzpicture}}
	\hspace{2cm}
	\subfloat[\label{fig:perf}Performances]{
		\centering
		\begin{tikzpicture}
			\draw[->] (0,0) -- (3,0)node[below, pos = 0.9] {$t$};
			\draw[->] (0,0) -- (0,2) node[left, pos = 0.8] {\rotatebox{90}{data}};
			\draw[red] (0,1) -- (3,1.7) node[below = -0.05cm, pos = 0.95] {$\alpha$};
			\draw[blue] (0,0) -- (1,0) -- (3,2) node[below, pos = 0.5] {$\beta$};
			\draw[green!70!black, <->, thick] (0,1) -- (2,1) node[below = -0.1cm, pos = 0.2] {$d_{\max}$};
			\draw[yellow!70!black, <->, thick] (1,0) -- (1, 1.25) node[right = -0.1cm, pos = 0.5] {\rotatebox{90}{$b_{\max}$}};
			\draw[red!70!black, <->, thick] (0,1.65) -- (2.7, 1.65) node[above = -0.1cm, pos = 0.5] {$\ell_{\max}$};
	\end{tikzpicture}}
	\caption{Processes and worst-case performances.}
\end{figure}
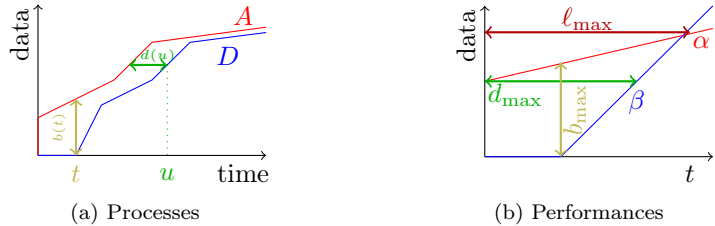

\paragraph{From performances to output arrival curves}

It is also possible to compute alternative arrival curves of the output processes using delay and backlog upper bounds of the servers. 

\begin{theorem}[{\cite[Theorem 7.4]{BBL18}}]
	\label{th:delay-curve}
	Consider a FIFO server crossed by an $\alpha$-constrained flow. Suppose that $d$ is an upper bound of the delay for this flow. Then $\alpha\deconv\delta_{d}$ is an arrival curve for the departure process.
\end{theorem}

\begin{theorem}
	\label{th:backlog-curve}
	Consider an $\alpha$-constrained flow crossing a system, with $\alpha = \gamma_{b,r}$. Assume that the last server of the system crossed by the flow offers a continuous service curve $\beta_n$ and that $\alpha$ is the only constraint for the flow. If $B$ is the maximum backlog in the whole system for this flow, then $\alpha'= \gamma_{B, r}$ is an arrival curve for the departure process. 
\end{theorem}

\begin{proof}
	A similar result has already been proved in~\cite{Bou19} in a slightly different setting (strict service curves). For the sake of completeness, we give the proof here, but it follows the lines of the previous proof.
	
	Let us denote by $A$ the cumulative arrival process of the flow and $D$ its departure process. Fix $s$ and $t$ such that $s\leq t$. One wants to show that $D(t)-D(s) \leq B + rt$.
	
	Let us first transform the arrival process $A$ in $A'$ so that $A'(u) = A(u)$ for all $u\leq s$ and $A'(u)$ is maximized for all $u>s$. As $\alpha$ is the only constraint for the flow, there exists $H\geq 0$ such that for all $u>t$, $A'(u) = A(s) + H  + r(u-s)$.
	
	If $D(s) = D(s+)$ (the departure process is right-continuous at $s$), then the backlog at time $s+$ with the modified arrival process $A'$ is $A'(s+) - D(s+) = A'(s) + H - D(s) \leq B$. Consequently, 
	$$D(t) - D(s) \leq A(t) - D(s) \leq A'(t) - D(s) \leq A(s) + H + r(t-s) - D(s) \leq B + r(t-s).$$
	
	In the case $D$ is not right-continuous at time $s$, one needs to modify the departure process to finish the proof. 
	Let $A_n$ (resp. $D_n$) be the aggregated arrival (resp. departure) process of the last server visited by the flow of interest. 
	There exists $v\leq s$ such that $D_n(t) \geq A_n(v) + \beta_n(s-v)$. 
	As $\beta_n$ is continuous, $D_n$ can be modified from time $s$ so that is continuous: take $D'_n(u) = \min(D_n(u), A_n(v) + \beta_n(s-v))$. Remark that as $v\leq s$, $A_n(v) = A'_n(v)$, so this new departure process is admissible and continuous on an interval $[s, s+\epsilon]$ with $\epsilon>0$. As $D'_n$ is continuous on $[s, s+\epsilon]$, the individual flows are also continuous, and in particular, $D'(s+) = D(s)\leq D(s+)$. We can then write
	$$D(t) - D(s) \leq A'(t) - D'(s+) \leq A(s) + H + r(t-s) - D(s) \leq B + r(t-s).$$
\end{proof}

\subsection{Network model}
Consider a network $\cN$ composed of $n$ servers numbered from 1 to $n$ and
crossed by $m$ flows named $f_1,\ldots,f_m$, such that
\begin{itemize}
	\item each server $j$ guarantees a service curve
	$\beta_j$ and has a greedy shaper $\sigma_j$. The service policy is FIFO;
	\item each flow $f_i$ is $\alpha_i$-constrained and circulates along an acyclic path $\pi_i= \langle \pi_i(1),\ldots,\pi_i(\ell_i)\rangle$ of length $\ell_i$.
\end{itemize}
We will always assume in the following that 
arrival curves and greedy shapers are token-bucket and the service curves rate-latency. We will use the following additional notations:

\begin{itemize}
	\item $F^{(j)}_i\in\cF$ is the cumulative process of flow $i$ entering server $j$. The departure process after the last server crossed by flow $f_i$ is be denoted $F^{(n+1)}_i$;
	\item the arrival curve of $F_i^{(j)}$ is denoted $\alpha_i^{(j)}:t\mapsto b_i^{(j)} + r_i t$. In particular, $F_i^{(\pi(1))}$ is $\alpha$-constrained and $b_i^{(\pi(1))} = b_i$;
	\item the service curve of server $j$ is $\beta_j:t\mapsto R_j(t-T_j)_+$ and the shaping curve is $\sigma_j:t\mapsto L_j + C_j t$;
	\item for a server $j$, we define $\fl(j) = \{i~|~ \exists \ell,~\pi_i(\ell)  =j\}$ the set of indices of the flows crossing server $j$ and $\fl(h, j) = \{i~|~ \exists \ell,~(\pi_i(\ell), \pi_i(\ell+1)  =(h, j )\}$  the set of indices of the flows crossing servers $h$ and $j$ in sequence;
	\item for all flows $f_i$, for $j\in\pi(i)$, we denote by $\su_i(j)$ is the successor of server $j$ in flow $f_i$. If $j = \pi(\ell_i)$, then $\su_i(j) = n+1$. For all servers $j$, $\pr(j)$ is the set of predecessors of server $j$.    
\end{itemize}

We call the family of cumulative $(F^{(j)}_i)_{i\in\sint{m}, j\in\pi_i \cup \{n+1\}}$ a trajectory of the network, and an admissible trajectory if it satisfied all the network calculus constraints described above: arrival, service 
shaping and FIFO constraints. 

 The induced graph $G_{\cN} = (\sint{n},\A)$ is the
directed graph whose vertices are the servers and the set of arcs is
$$\A = \{(\pi_i(k),\pi_i(k+1))~|~i\in\sint{m}, k\in\sint{\ell_i-1}\}.$$
As we will focus on the
performances in server $n$ or of a flow ending at server $n$, we can assume without loss of generality that the network is connected and has a unique
final strictly connected component, which contains $n$.

\begin{itemize}
	\item If the induced graph $G_{\cN}$ is a line, we say that the network is a tandem network;
	\item if the induced graph $G_{\cN}$ is a tree, we say that the network is a tree network (as the network is assumed to be connected and have a unique final component, all maximal paths end at node $n$, that is the unique sink of the network);
	\item if the induced graph $G_{\cN}$ is acyclic, we say that the network is feed-forward;
	\item if the induced graph $G_{\cN}$ contains cycles, we say that the network has cyclic dependencies (or is not feed-forward).
\end{itemize}

\paragraph{Stability} We will also be interested in the network stability. 

\begin{definition}[Global stability]
	A network  is {\em globally stable} if the backlogged periods of each server are uniformly bounded. 
\end{definition}

Deciding if a network is stable is an open problem in general, and only partial results exist. A necessary condition is that the arrival rate in each server is less than the service rate, but this condition is not sufficient: it has been shown in~\cite{And2007} that there exists networks with arbitrary small {\em local} load that can be unstable. 

Local stability refers to the arrival rate being less than the service rate in every server of the network. In the following, we will always assume local stability. In our setting, this means that for all server $j$, $\sum_{i\in\fl(j)} r_i \leq R_j$.

\section{State of the art on computing bounds in FIFO networks in NC}
\label{sec:soa}

We only describe techniques for feed-forward networks. For cyclic dependencies, those techniques must be used with the fix-point for example. This latter point will be developed in Section~\ref{sec:cyclic}.

\subsection{TFA (Total flow analysis) and TFA++}
TFA and TFA++ as described here can be used exclusively for FIFO networks. It is based on Theorem~\ref{th:delay-curve} and the fact that the worst-case delay in a FIFO server is the same for all flows crossing it.

\begin{algorithm}[htbp]
	\caption{TFA analysis: delay of flow $i$}
	\label{algo:tfa}
	\Begin{
	\ForEach{server $j$ in the topological order}
	{
		$b \gets \sum_{i\in\fl(j)} b_i^{(j)}$\;
		$d_j \gets T_j + \frac{b}{R_j}$\;
		$b_i^{(\su_i(j))} \gets b_i^{(j)} + r_id_j$
	}
	{\bf return } $\sum_{j\in \pi_i} d_j$
}
\end{algorithm}

TFA++ is similar to TFA except that  it is taking into account the maximum service rate (as a greedy-shaper) of the preceding servers. It has first been introduced in Grieux's PhD thesis~\cite{Gri04} and then popularized under the name TFA++ in~\cite{ML17}. In short, between Algorithm~\ref{algo:tfa} and~\ref{algo:tfa++}, lines 3 and 4 differ. 
The case with cyclic dependencies is studied in~\cite{TBM19}, and will be commented in Section~\ref{sec:cyclic}.

\begin{algorithm}[htbp]
	\caption{TFA++ analysis: delay of flow $i$}
	\label{algo:tfa++}
	\Begin{
		\ForEach{server $j$ in the topological order}
		{
			$\alpha \gets \sum_{h\in\pr(j)} \min(\sigma_h, \sum_{i \in \fl(h, j)}\alpha_i) + \sum_{i~|~\pi_i(1) = j}\alpha_i$\;
			$d_j\gets d_{\max}(\alpha, \beta_j)$\;
			$b_i^{(\su_i(j))} \gets b_i^{(j)} + r_id_j$
		}
		{\bf return } $\sum_{j\in \pi_i} d_j$
	}
\end{algorithm}

\subsection{SFA (Separated flow analysis)}
SFA is the technique that uses the pay burst only once phenomenon through the use of (min, plus) operators. We give here a possible algorithm when the network is FIFO, by choosing a particular value of $\theta$ in Theorem~\ref{th:fifo}. This choice is locally optimal: from Theorem~\ref{th:backlog-curve}, the backlog bound  characterizes the maximum burst of the output arrival curves, so $\theta$ is chosen so as to minimize the backlog bound for each flow at each server.

\begin{corollary}[of Theorem~\ref{th:fifo}]
	Consider a FIFO server with service curve $\beta:t\mapsto R(t-T)_+$, crossed
	by two flows $f_1$ and $f_2$ with respective arrival curves $\alpha_1:t\mapsto b_1 + r_1 t$ and $\alpha_1:t\mapsto b_1 + r_1 t$. A residual service curve for flow $f_1$ is 
	$$\beta': t\to (R-r_2)(t - (T + b_2/R)_+.$$ 
	The output arrival curve is $\alpha_1'= \alpha_1 + (T+ b_2/R)r_1.$		
\end{corollary}
This is Theorem~\ref{th:fifo} with $\theta = T + b_2/R$.

Algorithm~\ref{algo:sfa} describes the procedure to compute the delay of a flow with the SFA method.
\begin{algorithm}[htbp]
	\caption{SFA analysis: delay of flow $i_0$}
	\label{algo:sfa}
	\Begin
	{
		\ForEach{server $j$ in the topological order}
		{
			\ForEach{flow $i\in\fl(j)$}
			{
				$b \gets \sum_{k\in\fl(j)-i} b_k^{(j)}$\;
				$b_i^{(\su_i(j))} \gets b_i^{(j)} + (T_j + b/R_j)r_i$\;
				$T_i^{(j)} \gets  (T_j + b/R_j)r_i$\;
				$R_i^j \gets R_j - \sum_{k\in\fl(j)-i}r_k$
				
		}}
		{\bf return } $\sum_{j\in \pi_{i_0}} T_{i_0}^{(j)} + b_{i_0} / (\min_{j\in\pi_{i_0}} R_{i_0}^j)$
	}
\end{algorithm}

\subsection{Deborah}
Deborah(DElay BOund Rating AlgoritHm)~\cite{BLMS2010} is a software designed to compute delay bounds in FIFO tandem networks. It is based on the optimization of $\theta$ parameters that appear in Theorem~\ref{th:fifo} for each server. These parameters can in particular be optimized when the flows are nested (each flow is contained, contains and is disjoint from any other flow) or for sink-trees. It is showed in~\cite{LMMS2006} that the delay bounds are tight for sink trees, and in~\cite{BLMS2008} that even for very small networks, the bound is not tight for other tandem topologies. The general case of tandems can be tackled by cutting flows to make is nested~\cite{LMS2007, LMS2008}. 

The tool requires token-bucket arrival curves and rate-latency service curves, and does not take into account the shaping effect of a maximal service curve. The aim of this paper is to study general topologies and the shaping effect. This is why we will not include this tool in our comparisons. Comparisons with linear programming methods can be found in~\cite{BS12}.

\subsection{Linear programming}
The linear programming approach developed in~\cite{BS12, BS15} consists in writing the network calculus constraints as linear constraints. If the arrival curves are piece-wise linear concave and the service curves piece-wise linear convex, then the exact worst-case bounds can be computed by a MILP (Mixed-integer linear program). However this solution is very costly as the number of variables is exponential and there are integer variables. The MILP can be relaxed by removing the integer variables and there corresponding constraints. While this relaxation gives accurate bounds (better that other methods), the number of constraints is still too high to be able to compute bounds in large network. 

\medskip

In the following, we will compare our contribution with the TFA++, SFA and linear programming (LP). We call tractable or scalable TFA++ ans SFA, as their complexity enables the analysis of large networks, whereas we call untractable LP, due to its (super)-exponential complexity.
\section{A polynomial-size linear program with for tree networks}
\label{sec:tree}

In this section, we propose to modify the linear program of~\cite{BS12} to keep the number of constraints and variables polynomial in the size of the network, for a tree network. Simply removing constraints can make the bounds more pessimistic than SFA or TFA, and we will then propose to incorporate these bounds to improve the tightness of our new bound. We also adapt the linear program so that it can take into account the shaping of the cumulative processes due to the link capacities.

\subsection{A linear program to compute upper bound delays}
This paragraph describes the linear program obtained from the simplification of the linear program presented in from~\cite{BS12}. We write the variables in bold letters. 
To give the intuition of these variables and constraints, we apply the construction on the small network for Figure~\ref{fig:toy}. The linear program is given in Table~\ref{tab:toy}.
 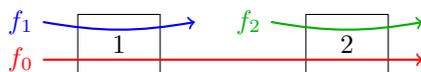
\begin{figure}[htbp]
 	\centering
 	\begin{tikzpicture}[server/.style={shape=rectangle,draw,minimum height=.8cm,inner xsep=3ex}]
 	\node[server,name=S1] at (0,0) {$1$};
 	\node[server,name=S2] at (3,0) {$2$};
 	\draw[->, thick, red] (-1,-0.2) node[left] {$f_0$} -- (4,-0.2);
 	\draw[->, thick, blue] (-1,0.3) node[left] {$f_1$} to[out=-10,in=-170] (1,.3);
 	\draw[->, thick, green!70!black]  (2,0.3) node[left] {$f_2$} to[out=-10,in=-170] (4,.3);
 	\end{tikzpicture}
 	\caption{Toy network for the linear program description.}
	\label{fig:toy}
 \end{figure}

For the general description of the linear program, we assume a tree-network, and that the flow we are interested in computing the delay ({\em flow of interest}) ends at the sink of the tree.  

Each server $j$ has a unique successor that we denote $\su(j)$, except node $n$ that is a sink, for which we set $\su(n) = n+1$. To simplify the notations, we will use $F_{i}^{(\su(\pi_i(\ell_i)))}$ instead of $F_{i}^{(n+1)}$ for the departure processes of the flows. 
Let us denote $d(j)$ the depth of server $j$. We set the depth of the sink $d(n) = 1$, and $d(j) = d(\su(j)) + 1$.

\paragraph{Time variables and constraints:}
We introduce one variable ${\bf t_0}$ representing the departure time of 
the bit of interest. 
For each server $j$, we reserve $d(j) + 1$ variables representing dates ${\bf t}_k$ for $k \in [u_{\min}(j),\ldots, u_{\max}(j)]$. The sets $[u_{\min}(j),\ldots,$ $ u_{\max}(j)]$ are disjoint. 
We set the following constraints (to be explained with the FIFO and service constraints):
\begin{itemize}
	\item for all $k \in [u_{\min}(j),\ldots, u_{\max}(j) - 1]$, ${\bf t}_k \geq {\bf t}_{k+1}$;
	\item for all $k\leq u_{\max}(j) - u_{\min}(j) - 1$, and with $h = \su(j)$, ${\bf t}_{k + u_{\min}(j)} \leq  {\bf t}_{k + u_{\min}(h)}$.
\end{itemize}

The total number of time variables is then at most $(n+2)(n+1)/2$ and the number of time constraints is at most $n(n+1)$. The worst case is obtained for the tree with maximal depth, that is the tandem networks.  

\begin{example}
	In our example, we have 6 time variables: $\bt_0$ for the exit time of the bit of interest, $\bt_1$ and $\bt_2$ are defined for server 2, and $\bt_3$, $\bt_4$, $\bt_5$ are defined for server 1. 
\end{example}

\paragraph{Process variables:} Now he have introduced the dates variables, we can introduce variables of type $\bF^{(i)}_j(\bt_k)$ for the cumulative processes. This corresponds to the value of the arrival cumulative process of flow $f_i$ at server $j$ and at time $t_k$. The variable $\bF^{(j)}_i(\bt_k)$ exists if $j\in \pi_i \cup \{\su(\pi_i(\ell_i))\}$ and $k \in [u_{\min}(j),\ldots, u_{\max}(j)]$.

For each flow, there is at most one process variable per time variable, so the number of constraints is at most $m(n+2)(n+1)/2$ (this is the worst case where all flows cross all servers).

\paragraph{FIFO constraints} For each time $t$, there exists a date $s\leq t$ such that all data exited by time $t$ have arrived by time $s$. The FIFO policy ensures that this is also true for all flows crossing that server. When applied at the specific dates, we obtain the following constraints:
for all server $j$, and $h = \su(j)$, for all $k\leq u_{\max}(h) - u_{\max}(h)$, 
\begin{itemize}
	\item ${\bf t}_{k + u_{\min}(j)} \leq  {\bf t}_{k + u_{\min}(h)}$ (as mentioned above);
	\item for all $i\in\fl(j)$,  $\bF_i^{(j)}(\bt_{k + u_{\min}(j)}) =  \bF_i^{(h)}(\bt_{k + u_{\min}(h)})$.
\end{itemize}

For a flow crossing server a server of depth $d$, there are $d$ FIFO constraints. Then there are in total at most $mn(n+1)/2$ FIFO constraints.

\paragraph{Service constraints} This is the point where the simplification is done compared to~\cite{BS12}. We introduce only one constraint per server, at time $t_{u_{\max}(\su(j))}$: we apply the formula $\sum_{i\in\fl(j)} F^{(\su(j))} \geq \sum_{i\in\fl(j)} F^{(j)}_i \conv \beta_j$, and obtain by definition of the (min, plus)-convolution: for all server $j$, denoting $h = \su(j)$, 
\begin{itemize}
	\item $\sum_{i\in \fl(j)} \bF^{(h)}_i(\bt_{u_{\max}(h)}) \geq \sum_{i\in \fl(j)} \bF^{(j)}_i(\bt_{u_{\max}(j)})$ and  
	\item $\sum_{i\in \fl(j)} \bF^{(h)}_i(\bt_{u_{\max}(h)}) \geq\sum_{i\in \fl(j)} \bF^{(j)}_i(\bt_{u_{\max}(j)}) + R_j(\bt_{u_{\max}(h)} - \bt_{u_{\max}(j)}) - R_jT_j$.
\end{itemize}

These constraints, together with the FIFO constraints and the monotony of the cumulative processes also impose an order on the dates: 
for all $k \in [u_{\min}(j),\ldots, u_{\max}(j) - 1]$, ${\bf t}_k \geq {\bf t}_{k+1}$.

There are two service constraints per server, so $2n$ in total.

\paragraph{Arrival constraints}
For all flow $f_i$ one can add the constraints built from its arrival curve. Let $j$ be the first server crossed by flow $f_i$. 
for all $u_{\min(j)}\leq u<v \leq u_{\max(j)}$, we have 
$$\bF_i^{(j)}(\bt_u) - \bF_i^{(j)}(\bt_v) \leq b_i + r_i(\bt_u - \bt_v).$$

If a flow arrives at a server of depth $d$, it induces $d(d+1)/2$ constraints. There are then at most $mn(n+1)/2$ arrival constraints. 

\paragraph{Monotony constraints} We havw defined an total order for the dates of the cumulative arrival process for each flow (arriving at its first server). The arrival processes are non-decreasing, and we can translate this into linear constraints. 
Due to the FIFO constraints and service constraints, there is no need to write these types of constraints them for the arrival processes at each server.  For each flow $f_i$, let $j$ be the first server it crosses, for all $k \in [u_{\min(j)}, u_{\max(j)}-1]$, we have 
$$\bF_i^{(j)}(\bt_k) \geq  \bF_i^{(j)}(\bt_{k + 1}).$$ 

For each flow crossing a server of depth $d$, there are $d$ monotony constraints, so $mn(n+1)/2$ in total. 

\paragraph{Delay objective}
To obtain an upper bound of the worst-case delay of flow $f_i$ ending at server $n$ is $\max: t_0 - t_{u_{\min}(j)}$ where $j$ is the first server crossed by flow $f_i$. 

\paragraph{Backlog objective} Alternatively, to obtain an upper bound of the worst-case backlog of flow $f_i$ starting at server $j$ and ending at server $n$ in the network, one introduces the following constraints and objective: 
\begin{itemize}
	\item for all $u_{\min(j)} \leq u\leq u_{\max(j)}$, $\bF_i^{(j)}(\bt_0) - \bF_i^{(j)}(\bt_u) \leq b_i + r_i(\bt_0 - \bt_u)$;
	\item $\max~:~  \bF_i^{(j)}(\bt_0) -  \bF_i^{(n+1)}(\bt_0)$. 
\end{itemize}

\begin{theorem}
	\begin{enumerate}
	\item Let $D$ be a solution of the linear program described above and $d$ be the worst-case delay of the flow of interest. Then $D\geq d$. 
	\item Let $B$ be a solution of the linear program described above and $b$ be the worst-case backlog of the flow of interest. Then $B\geq b$. 
	\end{enumerate}
\end{theorem} 

\begin{proof}
	The linear program just described above has a subset of constraints of the one described in~\cite{BS12} for computing the exact worst-case delay, with the same objective. Then the results of our linear constraint is then an upper bound of it, hence the result.
\end{proof}

Compared to the approach of~\cite{BS12}, we removed most of the service constraints (only one is kept per server). The number of time variables is then reduces from $2^{n+1}-1$ for a tandem with $n$ servers to $(n+1)(n+2)/2$ variables, and the number of constraints is now $O(mn^2)$. 
Moreover, all the time variables related to one server are naturally ordered with the FIFO and service constraints. As a consequence, there is no need to introduce Boolean variables to ensure the monotony of the processes. 
However, the following example shows that the performances can get too pessimistic compared to the scalable methods.  

\begin{table}[htbp]
	\centering
	\begin{tabular}{|ll|}
		\hline
		Maximize: & $\bt_0 - \bt_3$  \\
		\hline
		 such that &$\bt_3\leq \bt_1 \leq \bt_0$\\
		 (Time &$\bt_4\leq \bt_2$\\
		  constraints)&$\bt_2\leq \bt_1$\\
		 &$\bt_5 \leq \bt_4\leq \bt_3$\\
		 \hline
		 (FIFO & $\bF_0^{(1)}(\bt_3) =\bF_0^{(2)}(\bt_1) =\bF_0^{(3)}(\bt_0)$\\  
		 constraints) &$\bF_0^{(1)}(\bt_4) =\bF_0^{(2)}(\bt_2)$\\
		 &$\bF_1^{(1)}(\bt_3) =\bF_1^{(2)}(\bt_1)$\\  
		 &$\bF_1^{(1)}(\bt_4) =\bF_1^{(2)}(\bt_2)$\\
		 &$\bF_2^{(2)}(\bt_1) =\bF_2^{(3)}(\bt_0)$\\
		 \hline
		 (Service& $\bF^{(2)}_0(\bt_2) +  \bF^{(2)}_1(\bt_2) \geq \bF^{(1)}_0(\bt_5) +  \bF^{(1)}_1(\bt_5) + R_1(\bt_{2} - \bt_{5}) - R_1T_1 $ \\ 
		 constraints)& $\bF^{(2)}_0(\bt_2) +  \bF^{(2)}_1(\bt_2) \geq \bF^{(1)}_0(\bt_5) +  \bF^{(1)}_1(\bt_5)$ \\ 
		 &$\bF^{(3)}_0(\bt_0) +  \bF^{(3)}_2(\bt_0) \geq \bF^{(2)}_0(\bt_2) +  \bF^{(2)}_2(\bt_2) + R_2(\bt_{0} - \bt_{2}) - R_2T_2 $\\
		 &$\bF^{(3)}_0(\bt_0) +  \bF^{(3)}_2(\bt_0) \geq \bF^{(2)}_0(\bt_2) +  \bF^{(2)}_2(\bt_2)$\\
		 \hline
		 (Arrival& $\bF_0^{(1)}(\bt_3) - \bF_0^{(1)}(\bt_4) \leq b_0 + r_0(\bt_3 - \bt_4)$\\
		constraints)& $\bF_0^{(1)}(\bt_4) - \bF_0^{(1)}(\bt_5) \leq b_0 + r_0(\bt_4 - \bt_5)$\\
		& $\bF_0^{(1)}(\bt_3) - \bF_0^{(1)}(\bt_5) \leq b_0 + r_0(\bt_3 - \bt_5)$\\
		& $\bF_1^{(1)}(\bt_3) - \bF_1^{(1)}(\bt_4) \leq b_1 + r_1(\bt_3 - \bt_4)$\\
		& $\bF_1^{(1)}(\bt_4) - \bF_1^{(1)}(\bt_5) \leq b_1 + r_1(\bt_4 - \bt_5)$\\
		& $\bF_1^{(1)}(\bt_3) - \bF_1^{(1)}(\bt_5) \leq b_1 + r_1(\bt_3 - \bt_5)$\\
		& $\bF_2^{(2)}(\bt_1) - \bF_i^{(2)}(\bt_2) \leq b_2 + r_2(\bt_1 - \bt_2)$\\
		\hline
		(Monotony& $\bF_0^{(1)}(\bt_3) \geq \bF_0^{(1)}(\bt_4) \geq\bF_0^{(1)}(\bt_5)$\\
		constraints)& $\bF_1^{(1)}(\bt_3) \geq \bF_1^{(1)}(\bt_4) \geq\bF_1^{(1)}(\bt_5)$\\
		& $\bF_2^{(2)}(\bt_1) \geq \bF_2^{(2)}(\bt_2)$ \\
		\hline
	\end{tabular}
\caption{Linear program from the simplification of~\cite{BS12} for the toy example of Figure~\ref{fig:toy}.}
\label{tab:toy} 
\end{table}

\begin{example}
	Consider the example of Figure~\ref{fig:toy}, with arrival curves $\alpha: t\mapsto 1+ t$ for all flows and service curves $\beta: t\to 4(t-1)_+$ for both servers. Assume moreover that server 1 has the maximum service curve $\beta^u_1: t\mapsto 4t$. 
	
	The delay obtained with the SFA or LP method is 2.83, with TFA++, 2.95 and with this new linear program is 3.25. The reason for this is that a service constraint for the first server has been removed compared to the linear program of~\cite{BS12}: the time variable $t_1$, used to describe the flows entering the second server appears only as a FIFO constraints in server 1, and is not involved in a service constraint. In this linear program, it is set to $t_0$, inducing a larger delay. We see that $t_1 = t_0$. All data have been served for flow 2 before serving flow 0, as if server 2 gave the priority to flow 2. 
	
	Figure~\ref{fig:toy-lp} shows the trajectories computed by the linear program. 
	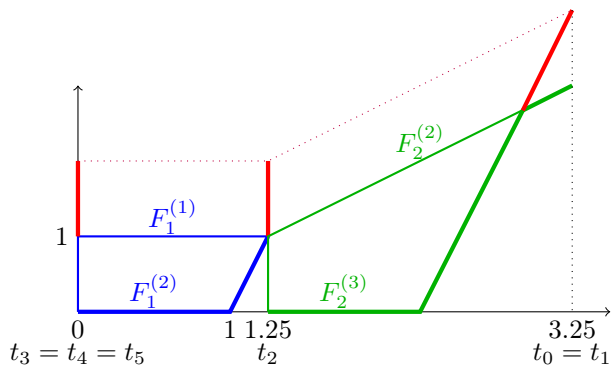
\begin{figure}[htbp]
		\centering
		\begin{tikzpicture}
		\draw[->] (0,0) -- (7, 0);
		\node[below] at (2, 0) {1};
		\node[left] at (0, 1) {1};
		\node[below] at (0, 0) {0};
		\node[below = 0.3cm] at (0, 0) {$t_3 = t_4 = t_5$};
		\node[below] at (2.5, 0) {1.25};
		\node[below = 0.3cm] at (2.5, 0) {$t_2$};
		\node[below] at (6.5, 0) {3.25};
		\node[below = 0.3cm] at (6.5, 0) {$t_0 = t_1$};
		\draw[->] (0,0) -- (0, 3);
		\draw[thick, blue] (0, 0) -- (0, 1)  -- (2.5, 1) node[pos = 0.5, above=-0.1cm] {$F^{(1)}_1$};
		\draw[ultra thick, red] (0, 1) -- (0, 2);
		\draw[ultra thick, blue] (0,0) -- (2, 0) node[pos = 0.5, above = -0.1cm] {$F^{(2)}_1$} -- (2.5, 1);
		\draw[ultra thick, red] (2.5, 1) -- (2.5, 2);
		\draw[thick, green!70!black] (2.5, 0) -- (2.5, 1) -- (6.5, 3) node[pos = 0.5, above=-0.1cm] {$F^{(2)}_2$};
		\draw[ultra thick,green!70!black ] (2.5, 0) -- (4.5, 0) node[pos = 0.5, above=-0.1cm] {$F^{(3)}_2$}-- (5.85, 2.67) -- (6.5, 3);
		\draw[red, ultra thick] (5.85, 2.67) -- (6.5, 4);
		\draw[black, dotted] (6.5, 4) -- (6.5, 0);
		\draw[purple, dotted] (0, 2) -- (2.5, 2) -- (6.5, 4);
		\end{tikzpicture}
		\caption{Trajectory reconstructed from the toy example. (blue) cumulative processes of flow 1 at server 1; (green) cumulative processes of flow 2 at server 2;(red) cumulative addition of the processes of flow 0. At time 0, the burst of size 1 arrives. It is transmitted at time 1.25, and served until time 3.25. }
		\label{fig:toy-lp}
	\end{figure}

\end{example}

\subsection{Adding SFA, TFA++ and shaping constraints}
We have just seen an example where the linear program just described  computes performance bounds that are more pessimistic than with the SFA  or TFA++. In this paragraph, we will see that adding more constraints computed with the SFA and TFA++ can drastically improve the bounds. Moreover, we show that we can also incorporate shaping constraints in our linear program. 

\paragraph{TFA++ and SFA constraints}
The intuition is the following: compared to the linear program of~\cite{BS12}, we have removed many service constraints, that were necessary to retrieve the exact worst-case delay. The idea here is to replace these service curve constraints by pure delay curves using Theorem~\ref{th:delay-curve}. SFA delays for each flow can also generalize this idea for a sequence of servers. 

The TFA++ algorithm from~\cite{ML17} computes for each server $j$ an upper bound of its delay $d_j^{TFA}$, that is satisfied for each flow (as the service policy is FIFO). 

One can replace each FIFO constraint setting an order between dates ${\bf t}_{k + u_{\min}(j)} \leq  {\bf t}_{k + u_{\min}(h)}$ by 
${\bf t}_{k + u_{\min}(j)} \leq  {\bf t}_{k + u_{\min}(h)} \leq {\bf t}_{k + u_{\min}(j)} + d_j^{TFA}$.  

We add as many constraints as FIFO constraints, that is at most $mn(n+1)/2$.

 The SFA algorithm computes the delay for each flow. Let  $d_i^{SFA}$ be the delay of flow $f_i$ computed with this method. 
Let $j$ and $h$ be respectively the first server and the successor of the last server crossed by flow $i$. For all $k \in [u_{\min(h)}, u_{\max(h)} - 1]$, by using successive FIFO constraints, we have 
$$\bF_i^{(h)}(\bt_{u_{\min}(h) + k}) = \bF_i^{(j)}(\bt_{u_{\min}(j) + k}).$$ As these represent the arrival and departure in/from the system of a bit of data of flow $f_i$. One can add the constraint
$$\bt_{u_{\min}(h) + k} - \bt_{u_{\min}(j) + k} \leq d_i^{SFA}.$$
We add at most $n$ constraints per flow, that is a total of at most $mn$.

\paragraph{Shaping constraints}

From Theorem~\ref{th:output}, we now that the aggregate process at the departure of node $j$ is constrained by the token-bucket $\sigma_j$. One can then add the following constraints: 
for each server $j$, let $h = \su(j)$. For all $u_{\min}(h) \leq u < v \leq u_{\max}(h)$, 
$$\sum_{i\in \fl(j)} \bF^{(h)}_i(\bt_{u}) - \sum_{i\in \fl(j)} \bF^{(h)}_i(\bt_v) \leq  L_j + C_j(\bt_u - \bt_v). $$

The number of additional variables is $d(d+1)/2$ for a node of depth $d$, so there are at most $n^2(n+1)/2$ variables.

The additional variables for the toy example of Figure~\ref{fig:toy} is given in Table~\ref{tab:toy-add}.

\begin{table}[htbp]
	\centering
	\begin{tabular}{|ll|}
		\hline
		such that & $\bt_1 - \bt_3 \leq d^{TFA}_1$ \\
		(TFA++& $\bt_2 - \bt_4 \leq d^{TFA}_1$ \\
		constraints)& $\bt_0 - \bt_1 \leq d^{TFA}_2$ \\
		\hline
		(SFA& $\bt_0 - \bt_3 \leq d^{SFA}_0$ \\
		constraints) & $\bt_1 - \bt_3 \leq d^{SFA}_1$  \\
		& $\bt_2 - \bt_4 \leq d^{SFA}_1$  \\
		& $\bt_0 - \bt_1 \leq d^{SFA}_2$  \\
		\hline
		(Greedy-shaper& $\bF_0^{(1)}(\bt_1) - \bF_0^{(1)}(\bt_2) \leq L_1 + C_1(\bt_1 - \bt_2)$\\
		constraints) & \\
		\hline
	\end{tabular}
\caption{Additional linear program for incorporating TFA++, SFA and greedy-shaper constraints for the toy example of Figure~\ref{fig:toy}.}
\label{tab:toy-add}
\end{table}

\begin{theorem}
	The objective of the linear program with the additional constraints is an upper bound of the worst-case delay (resp. backlog) of the flow of interest.
\end{theorem}

\begin{proof}
	The proof is similar to the proof in~\cite{BS12, BS15} (upper bound part). 
	Let $(F_{i}^{(j)})_{i, j}$ be an admissible trajectory for the network. 
	Let $t_0$ be the departure date (at server $n$) of the bit of data satisfying the worst-case delay (or backlog) of the flow of interest. 
	In the whole proof, we write $h = \su(j)$. We have $u_{\min(n+1)} = u_{\max(n+1)} = 0$. If $t_u$ are defined for $u\in [u_{\min(h)}, u_{\max(h)}]$, for some server $j$, one can define $s$ such that $i \in \fl(j)$, $F_i^{(j)}(s) = F_i^{(h)}(t_u)$. This is possible since the policy is FIFO. We denote $s = t_{u + 1 - u_{\min(h)} + u_{\min(j)}}$. 
	Moreover, there exists $s$ such that $\sum_{i\in\fl(j)} F_i^{(h)}(t_{u_{\min(h)}}) \geq \sum_{i\in\fl(j)} F_i^{(j)}(s) + \beta_j(t_{u_{\min(h)}} - s)$. We define $t_{u_{\min(j)}}$ as the minimum value of $s$ satisfying this service constraint. The minimum exists due to the continuity of $\beta_j$ and left-continuity of the processes.  
	
	Let us now fix the variables of the linear program:
	\begin{itemize}
		\item for all $u$, $\bt_u = t_u$;
		\item for all $i \in\fl(n)$, $\bF_i^{(n+1)}(\bt_0) = F_i^{(n+1)}(t_{0})$.
		\item for all $j$ and all $i\in\fl(j)$, for all $k\leq u_{\max(h)} - u_{\max(h)}$, $\bF_i^{(j)}(t_{u_{\min(j)}+k}) = \bF_i^{(h)}(\bt_{u_{\min(h)}+k})$;
		\item for all $j$ and all $i\in\fl(j)$, $\bF_i^{(j)}(\bt_{u_{\min(j)}}) = F_i^{(j)}(t_{u_{\min(j)}})$.
	\end{itemize} 
	
	One can remark, due to possible discontinuities of the cumulative processes that where defined, $\bF_i^{(j)}(\bt_u) \in [F_i^{(j)}(t_u) , F_i^{(j)}(t_u+)]$. 
	
	Now, one only need to check that the variables set this way satisfy all the linear constraints. By construction, the FIFO and service constraints are satisfied.
	
	As the system is causal, that is, $F_i^{(j)} \geq F_i^{(h)}$, and the cumulative processes are non-decreasing, the time and monotony constraints are satisfied.  
	
	For all $i$, let us denote $j = \pi_i(1)$ the first server crossed by flow $i$, for all $u_{\min(j)} \leq u < v \leq u_{\max(j)}$, 
	$\bF_i^{(j)}(\bt_u) - \bF_i^{(j)}(\bt_v) \leq F_i^{(j)}(t_u+) - F_i^{(j)}(t_v) \leq b_i + r_i(t_u-t_v)$. The arrival constraints are then satisfied. 
	
	Similarly, consider server $j$ and its departure processes $F_i^{(h)}$. For all $u_{\min(h)} \leq u < v \leq u_{\max(h)}$, 
	$$\sum_{i\in\fl(j)} (\bF_i^{(h)}(\bt_u) - \bF_i^{(h)}(\bt_v)) \leq \sum_{i\in\fl(j)}(F_i^{(h)}(t_u+) - F_i^{(h)}(t_v)) \leq L_j + C_j(t_u-t_v),$$ and the shaping constraints are satisfied. 
	
	Let us focus on the TFA++ constraints. For each FIFO constraint $\bF_i^{(j)}(\bt_u) = \bF_i^{(h)}(\bt_v)$, we have
	$$F_i^{(j)}(t_u) \leq \bF_i^{(j)}(\bt_u)  = \bF_i^{(h)}(\bt_v) \leq F_i^{(h)}(t_v+),$$ so $t_v - t_u \leq d_j^{TFA}$ and then the constraint $\bt_v - \bt_u\leq d_j^{TFA}$ is satisfied. 
	
	Similarly, for each flow $i$, let $j$ the first server it crosses and $h$ the successor of the last server it crosses. For all $t_v$ where $F_i^{(h)}$ is defined, and $t_u$ such that $\bF_i^{(j)}(\bt_u)  = \bF_i^{(h)}(\bt_v)$ (by transitivity) we have 
	$$F_i^{(j)}(t_u) \leq \bF_i^{(j)}(\bt_u)  = \bF_i^{(h)}(\bt_v) \leq F_i^{(h)}(t_v+),$$ so $t_v - t_u \leq d_i^{SFA}$ and then the constraint $\bt_v - \bt_u\leq d_i^{SFA}$ is satisfied. 
\end{proof}

\section{Linear programs for feed-forward networks}
\label{sec:ff}
The method proposed above strongly relies on the tree-topology of the network. In this section, we show how to extend the linear programming approach to feed-forward network. The first method is the more accurate and consists in unfolding the network in order to transform it into a tree. Unfortunately, the size of the tree might become exponential compared to the size of the original network. We then propose alternative constructions to reduce this complexity. For example the decomposition of the network into smaller pieces. These two constructions can of course be combined to optimize the trade-offs between accuracy and tractability, but this is out of the scope of this paper.
 
\subsection{Unfolding a feed-forward network into a tree}
Intuitively, this is equivalent to introducing FIFO and service date independently for each predecessor of servers. So if a servers has two successors,  FIFO and service dates for this server be will be introduced twice and independently. The unfolding of the network of Figure~\ref{fig:toy-ff} is depicted in Figure~\ref{fig:unfold}. 

\begin{figure}[htbp]
	\centering
	\begin{tikzpicture}
		[server/.style={shape=rectangle,draw,minimum height=.8cm,inner xsep=3ex}]
		\node[server,name=S0] at (0,0) {0};
		\node[server,name=S1] at (2,-1) {1};
		\node[server,name=S2] at (2,1) {2};
		\node[server,name=S3] at (4,0) {3};	
		\draw[thick, ->, red] (-1, -0.2) node [left] {$f_0$}-- (0, -0.2)-- (2, -1.2) -- (4, -0.2) -- (5, -0.2);
		\draw[thick, ->, blue] (-1, 0.2)  node [left] {$f_1$} -- (0, 0.2)-- (2, 1.2) -- (4, 0.2) -- (5, 0.2);
	\end{tikzpicture}
\caption{Toy feed-forward network.}
\label{fig:toy-ff}
\end{figure}
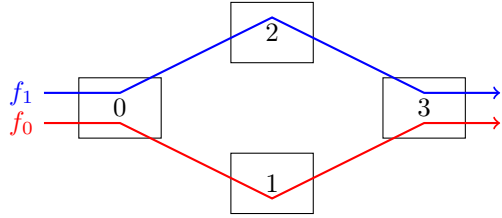

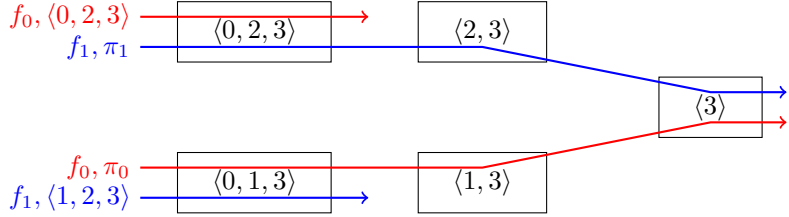
\begin{figure}[htbp]
	\centering
	\begin{tikzpicture}
	[ server/.style={shape=rectangle,draw,minimum height=.8cm,inner xsep=3ex}]
	\node[server,name=S01] at (0,-1) {$\langle 0, 1, 3 \rangle$};
	\node[server,name=S02] at (0,1) {$\langle 0, 2, 3 \rangle$};
	\node[server,name=S1] at (3,-1) {$\langle 1, 3 \rangle$};
	\node[server,name=S2] at (3,1) {$\langle 2, 3 \rangle$};
	\node[server,name=S3] at (6,0) {$\langle 3 \rangle$};
	\draw[thick, ->, red] (-1.5, -0.8) node [left] {$f_0, \pi_0$} -- (3, -0.8 )-- (6, -0.2) -- (7, -0.2);
	\draw[thick, ->, blue] (-1.5, -1.2) node [left] {$f_1, \langle 1, 2, 3\rangle$}-- (1.5, -1.2); 
	\draw[thick, ->, blue] (-1.5, 0.8) node [left] {$f_1, \pi_1$}-- (3, 0.8 )-- (6, 0.2) -- (7, 0.2);
	\draw[thick, ->, red] (-1.5, 1.2)node [left] {$f_0, \langle 0, 2, 3\rangle$} -- (1.5, 1.2);
	\end{tikzpicture}
	\caption{Unfolding of the network of Figure~\ref{fig:toy-ff}. For example, the original flow following the path $\langle 0, 1, 3\rangle$, leads to two flows in the unfolding: paths $\langle \langle 0, 1, 3\rangle, \langle 1, 3\rangle, \langle 3\rangle\rangle$ and $\langle \langle 0, 1, 3\rangle\rangle$.}
	\label{fig:unfold}
\end{figure}

\paragraph{The unfolding construction}
Consider a network $\cN$. Assume that this network is feed-forward, and that the servers are numbered such that if $(j,h)$ is an arc, then $j<k$, and that server $n$ is the only sink of the network.

The unfold-net of $\cN$ is denoted $\cU$ and defined as follows: \begin{itemize}
	\item let $\Pi$ be the set of paths in $\cN$ ending at $n$.
	Then $\Pi$ is the set of servers of $\cU$, server $\langle j_1, j_2 , \ldots, n\rangle$ offers a service curve $\beta_{j_1}$. 
	\item for all node $\pi = \langle j_1, j_2,\ldots , n\rangle$ and all flow $f_i$, let $\pi' = \langle j_1, j_2,\ldots , j_k\rangle$ the maximum common prefix of both $\pi_i$ and $\pi$. If this prefix is not empty, then there is a flow $(i, \pi)$ from node $\pi$ to $\langle j_k,\ldots , n\rangle$ with arrival curve $\alpha_i$. If flow $f_i$ is ending at server $n$, we call the flow from $\pi$ to $\langle n\rangle$ the copy for flow $f_i$.
\end{itemize}

One can easily check the following lemmas:
\begin{lemma} If $\pi \neq \pi'$, the two flows $(i,\pi)$ and $(i,\pi')$ do not share any common sub-path.
\end{lemma}

\begin{lemma}
	For each server $\pi = \langle j, \ldots, n\rangle$ of $\cU$, for all $i\in\fl(j)$ in $\cN$, there exists a flow $(i,\pi')$ such that $(i,\pi')\in\fl(\pi)$ in $\cU$.
\end{lemma}

\begin{theorem}
	Let $\cN$ be a feed-forward network and $\cU$ be its unfolding. Let $f_i$ a flow of $\cN$ ending at server $n$. Let $d_i^{\cN}$ be the worst-case delay of flow $f_i$ in $\cN$ and $d_i^{\cU}$ that of the copy of flow $i$ in $\cU$. Then $d_i^{\cN} \leq d_i^{\cU}$.  
\end{theorem}
\begin{proof}
	Let $(F^{(j)}_i)$ be an admissible trajectory for network $\cN$. Let us build a trajectory for $\cU$. For each flow $(i, \pi)$ of $\cU$, where $\pi = \langle j_1, j_2,\ldots , n\rangle$, if $\pi' = \langle j_1, j_2,\ldots , j_k\rangle$ is the maximum common prefix of both $\pi_i$ and $\pi$, we set $F_{i, \pi}^{(\langle j_x,\ldots , n\rangle)} = F_i^{(j_x)}$, and $F_{i, \pi}^{(n+1)} = F_i^{(j_{k+1})}$. We now need to check that this trajectory is admissible for $\cU$.
	\begin{itemize}
	\item First, the arrival processes $F_{i, \pi}^{(\pi)} = F_i^{(\pi_i(1))}$ so it is $\alpha_i$-constrained.
	\item Second, consider a server $\langle j, h, \ldots, n\rangle$. The processes arriving to (resp departing from) this server are $F_{i, \pi}^{(\langle j, h,\ldots, n\rangle)} = F_i^{(j)}$ (resp. $F_{i, \pi}^{(\langle h,\ldots, n\rangle)} = F_i^{(h)}$ or $F_{i, \pi}^{(n+1)} = F_i^{(h)}$) if $i\in \fl(j)$ and $\langle j, h,\ldots, n\rangle$ is a prefix of $\pi$. The arrival and departure processes are then the same as in server $j$ of $\cN$, so the service, shaping and FIFO constraints are all satisfied. 
	\end{itemize}  
	As a consequence, the trajectory is admissible in $\cU$. In addition, if we consider the copy $(i, \pi)$ of flow $f_i$, we have $F_{i, \pi}^{(\pi)} = F^{(\pi(1))}_i$ and $F_{(i, \pi)}^{(n+1)} = F_i^{(n+1)}$, so the delay for the copy of $f_i$ in $\cU$ is the same as the delay of flow $f_i$ in $\cN$. 
	
	For all admissible trajectory of $\cN$, we have built an admissible trajectory in $\cU$ with the same worst-case delay for flow $i$, which means that   $d_i^{\cN} \leq d_i^{\cU}$. 
\end{proof}

A similar result holds for the backlog bounds. This unfilding procedure can also be used for any methods defined on tree topologies, such as linear programming for blind multiplexing~\cite{BJT10} or ad-hoc algorithms~\cite{Bou19}. 

\subsection{Decomposition into a tree network by splitting flows}
\label{sec:split}
Another solution to split flows into smaller pieces in order to obtain a tree, or a forest (collection of trees), and compute the arrival curves at places flows have been cut. 
The splitting procedure has been described in~\cite{BBL18}, and we briefly recall it here. 

Consider $G_{\cN} = (\N_n, \A)$ the graph induced by $\cN$, and define $\A_r\subseteq A$ such that $(\N_n, \A - \A_r)$ is a tree or a forest. A flow $f_i$ i then transformed into flows $K_i$ flows $(f_i, k)$ with paths $\langle \pi_i(h^i_k),\ldots,  \pi_i(h^i_{k+1}-1)\rangle$ in $(\N_n, \A - \A_r)$, where  $h^i_1 = 1$ and  $(\pi_i(h^i_k), \pi_i(h^i_k+1)) \in \A_r$ for $1< k \leq  K_i$. The transformation is illustrated in Figure~\ref{fig:ff-decomp}.

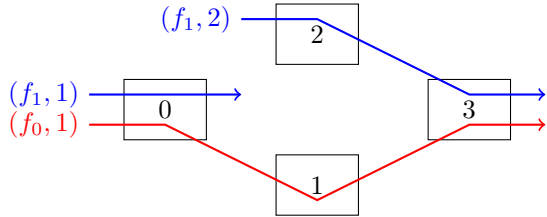
\begin{figure}[htbp]
	\centering
	\begin{tikzpicture}
	[server/.style={shape=rectangle,draw,minimum height=.8cm,inner xsep=3ex}]
	\node[server,name=S0] at (0,0) {0};
	\node[server,name=S1] at (2,-1) {1};
	\node[server,name=S2] at (2,1) {2};
	\node[server,name=S3] at (4,0) {3};
	\draw[thick, ->, red] (-1, -0.2) node [left] {$(f_0, 1)$} -- (0, -0.2)-- (2, -1.2) -- (4, -0.2) -- (5, -0.2);
	\draw[thick, ->, blue] (-1, 0.2) node [left] {$(f_1, 1)$} -- (1, 0.2);
	\draw[thick, ->, blue] (1, 1.2)  node [left] {$(f_1, 2)$}-- (2, 1.2) -- (4, 0.2) -- (5, 0.2);
	\end{tikzpicture}
	\caption{Decomposition of a network into a tree. Example of the network in Figure~\ref{fig:toy-ff} .}
	\label{fig:ff-decomp}
\end{figure}

Our aim is to compute a new network $\cN^{F}$ such that:
\begin{itemize}
	\item its induced graph is the forest $(\N_n, \A - \A_r)$;
	\item its servers offer the same guarantees as those of $\cN$;
	\item its flows are $\{(f_i, k)~|~k\in\{i, \ldots, K_i\}$. The arrival curve of $(f_i, k)$ is an arrival curve for $f_i$ at server $\pi_i(h^i_k+1)$;
	\item the arrival processes are shaped: for all $(j, j')\in\A^r$, flows $\{(f_i, k+1)~|~(\pi_i(h^i_k),\pi_i(h^i_k+1) ) = (j, j')\}$ are shaped by the curve $\sigma_j$.
\end{itemize} 
The arrival curves of the flows remain to be computed. 
As the network is feed-forward, the edges removed can be sorted in the topological order, and the computations be done according to this order, as described in Algorithm~\ref{algo:ff}. 

\begin{algorithm}
	\caption{Network analysis for feed-forward network by flow splitting}
	\label{algo:ff}
	\Begin{
		Sort $\A_r$ in the topological order in $\cN$ accoring to the first coordinate\;
		\ForEach{arc $(j,j')$ in the topological order}
		{\ForEach{flow $(f_i, k+1)$ starting at with $k>0$ starting at server $j'$}
			{Compute an arrival curve for flow $(f_i,k+1)$}}
	}
\end{algorithm}

To compute of the arrival curve for flow $(f_i,k+1)$, we will use Theorem~\ref{th:backlog-curve}. In short, to compute the arrival curve of flow $(f_i,k+1)$, one just have to compute the maximum backlog of flow $(f_i,k-1)$. 
Consider flows $(f_i,k)$ and $(f_i,k+1)$, 
$j$ the first server crossed by flow $(f_i, k)$.
There are two possibilities:
\begin{itemize}
	\item either we do not take into account the greedy shaper of server $j$, and Theorem~\ref{th:backlog-curve} can be applied directly. But not taking into account the shaping effect could lead to pessimistic bounds;
	\item or we take into account the greedy-shaper of server $j$. If done directly, Theorem~\ref{th:backlog-curve} cannot be applied, as it assumes that the arrival curve must be the only constraint of the flow. 
\end{itemize}

We propose to slightly transform the network so that the shaping effect can be taken into account for flows except flow $(f_i,k)$. The transformation is illustrated in Figure~\ref{fig:ff-shaping}. 

\begin{figure}[htbp]
	\centering
	\begin{tikzpicture}[scale=0.8, every node/.style={transform shape}]
	\node[ellipse, draw, minimum height=2cm,minimum width=4cm] at (0,0) {$\cN$};
	\draw[red, ->, ultra thick](-3, -0.2) -- (3, -0.2); 
	\draw[blue, ->, thick](-3, -0.4)  to[out=0,in=170] (3, -0.7);
	\draw[blue, ->, thick](-3, -0.6)  to[out=0,in=150]  (1.5, -1);
	\draw[purple, ->, thick](-3, 0.5) to[out=-10,in=-170] (3, 0.5); 
	\draw[purple, ->, thick](-3, 0.7) to[out=-10,in=-170] (3, 0.7);
	\draw[cyan, ->, thick](-2, 0.9) to[out=-10,in=-170] (2, 0.9);
	\draw[blue, ultra thick] (-3, -0.1) to[out=-10,in=-170] (-3, -0.7);
	\draw[purple, ultra thick] (-3, 0.4) to[out=-10,in=-170] (-3, 0.8);
	\node[purple] at (-3.3, 0.6) {$\sigma_2$};
	\node[blue] at (-3.3, -0.5) {$\sigma_1$};
	\end{tikzpicture}
	\hspace{1cm}
		\begin{tikzpicture}[scale=0.8, every node/.style={transform shape}]
	\node[ellipse, draw, minimum height=2cm,minimum width=4cm] at (0,0) {$\cN^{ns}$};
	\draw[red, ->, ultra thick](-3, -0.2) -- (3, -0.2); 
	\draw[blue, ->, thick](-3, -0.4)  to[out=0,in=170] (3, -0.7);
	\draw[blue, ->, thick](-3, -0.6)  to[out=0,in=150]  (1.5, -1);
	\draw[purple, ->, thick](-3, 0.5) to[out=-10,in=-170] (3, 0.5); 
	\draw[purple, ->, thick](-3, 0.7) to[out=-10,in=-170] (3, 0.7);
	\draw[cyan, ->, thick](-2, 0.9) to[out=-10,in=-170] (2, 0.9);
	\draw[blue, ultra thick] (-3, -0.3) to[out=-10,in=-170] (-3, -0.7);
	\draw[purple, ultra thick] (-3, 0.3) to[out=-10,in=-170] (-3, 0.7);
	\node[purple] at (-3.3, 0.6) {$\sigma_2$};
	\node[blue] at (-3.3, -0.6) {$\sigma_1$};
	\end{tikzpicture}
	\caption{System transformation: (left) the flow of interest (in red and blod) is shaped with two other flows (in blue) by the greedy-shaper $\sigma_1$; (right) the flow of interest is not shaped anymore. The rest of the network is not modified.}
	\label{fig:ff-shaping}
\end{figure}
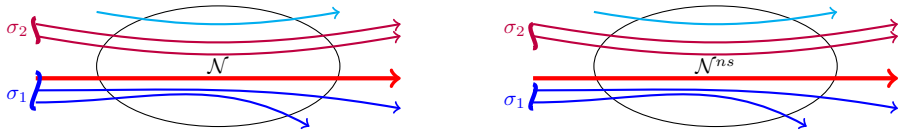

More precisely, let us consider $\cN_{(f_i, k)}$ the sub-network obtained from $\cN^F$ by keeping only the predecessors of the last node visited by flow $(f_i, k)$ and $\cN^{ns}_{(f_i, k)}$ the same network, except the shaping of the arrival processes becomes: for all $(j, j')\in\A^r$, flows $\{(f_{i'}, k'+1)~|~(\pi_{i'}(h^{i'}_{k'}),\pi_{i'}(h^{i'}_{k'}+1) ) = (j, j')\} \setminus \{(f_i, k)\}$ are shaped by the curve $\sigma_j$.
 
 \begin{lemma}
 	If $\alpha$ is an arrival curve for the departure process of flow $(f_i, k)$ in $\cN^{ns}_{(f_i, k)}$, then its is also an arrival curve for the departure process of that flow in $\cN_{(f_i, k)}$.
 \end{lemma}

\begin{proof}
	If flow $(f_i, k)$ was not shaped then the two networks are the same and there is nothing to show. Otherwise, 
	let us denote $I$ the set of flows shaped together with flow $(f_i, k)$ by $\sigma_j$ in  $\cN_{(f_i, k)}$. Let $(F_z^{(j)})$ be an admissible trajectory in $\cN_{(f_i, k)}$. It is then also an admissible trajectory in $\cN_{(f_i, k)}^{ns}$. To prove this, it is enough to check that $\sum_{z\in I\setminus \{(f_i, k)\}} F_j^{(j)}$ is $\sigma_1$-constrained: $\forall s\leq t$, 
	$$\sum_{z\in I\setminus \{(f_i, k)\}} F_z^{(j)}(t) - \sum_{z\in I\setminus \{f_i, k)\}} F_z^{(j)}(s) \leq \sigma_1(t-s) - (F_{(f_i, k)}^{(j)}(t) - F_{(f_i, k)}^{(j)}(s)) \leq \sigma_1(t-s), $$
	since $F_{(f_i, k)}^{(j)}$ is non-decreasing. 
	
	Then the set of possible departures processes of $(f_i, k)$ in $\cN_{(f_i, k)}$ is included in those of $(f_i, k)$ in $\cN_{(f_i, k)}^{ns}$, meaning that if $\alpha$ is an arrival curve for the departure process in $\cN_{(f_i, k)}^{ns}$, it is also one for those in $\cN_{(f_i, k)}$.
\end{proof}

According to Theorem~\ref{th:backlog-curve}, the arrival curve of flow $f_i(i,k)$ can then be computed according to Algorithm~\ref{algo:ff-flow}, where the backlog bound can be computed with the linear program given in Section~\ref{sec:tree}.

\begin{algorithm}
	\caption{Arrival curve for flow $(f_i, k+1)$}
	\label{algo:ff-flow}
	\Begin{
	\lIf{$k = 1$}{{\bf return} $\alpha_i$}
	Compute $B$ a backlog bound for flow $(f_i, k)$ in $\cN_{(f_i, k)}^{ns}$\;
	{\bf return} $\gamma_{B, r_i}$
}
\end{algorithm}

\section{Network with cyclic dependencies}
\label{sec:cyclic}
In this section, we study the case of networks with cyclic dependencies. For this, we will apply the fix-point analysis, that has already been described several times in~\cite{Chang2000, LT2001, BBL18}, to the analysis described in Section~\ref{sec:split}: edges are removed, so that the induced graph becomes a forest and flows are split accordingly. Because of the cyclic dependencies, removed edges cannot be sorted in the topological order, and the fix-point on the arrival curves of all splited flows has to be computed. 

We will first show in Section~\ref{sec:cyc-lp} that the fix-point can be computed using a linear program, and the uniqueness of this fix-point. In Section~\ref{sec:lp-tfa}, we will also apply this result to the TFA++ analysis for network with cyclic dependencies developed in~\cite{TBM19}.

\subsection{A linear program formulation for the fix-point analysis}
\label{sec:cyc-lp}
We formulate the fix-point equation with linear program. More precisely, we prove that the fix-point is obtained by extracting some variables from the optimal solution of a linear program.

We use the same notations as in Section~\ref{sec:split}, and denote $Z = \{(f_i, k),~i\in\{1, \ldots, m\}, k \leq K_i\}$ the set of flows in the network. As we will use Theorem~\ref{th:backlog-curve}, all the arrival curves computed for a flow $(f_i, k)$ will have arrival rate $r_i$, we only focus of the burst of the flows, that we denote $x_{(f_i, k)}$ to enforce that it is a variable.

For all $x = (x_z)_{z\in Z}$, and all $z\in Z$ let us define $\cL_z(x)$ as the backlog bound computed by Algorithm~\ref{algo:ff-flow}, when the burst parameter of flow $z'$ is $b_{z'}$ for all $z'\in Z$ and $\cL(x) = (\cL_z(x))_{z\in\cL}$.
 Theorem~\ref{th:fp} gives a sufficient condition for the stability of the network and the arrival curves of the split flows.

\begin{theorem}[{\cite[Theorem 12.1]{BBL18}}]
	\label{th:fp}
	If the maximal solution $x^*$ of $\cC = \{x \leq \cL(x)\}$ is finite $x^*$, then $\cN$ is globally stable and the burst of the arrival curve of flow $z$ is $x^*_z$.
\end{theorem}

\paragraph{Generic formulation for the fix-point equation of linear program}

Let us now focus on some properties of $\cL_z(x)$. The variables of this linear program are the time $\bt_u$ and process variables $\bF_i^{(j)}(\bt_u)$ described in Section~\ref{sec:tree}. The burst parameters only appear in the arrival curve constraints as $\bF_i^{(j)}(\bt_v) - \bF_i^{(j)}(\bt_u) \leq x_i + r_i(\bt_v - \bt_u)$. So if $x_i$ becomes a variable, the program remains linear. 

Then, there exists vectors a line-vector$A_z$ a column vector $C_z$ and a matrix $B_z$ such that 
$$\cL_z(x) = \max\{A_z(x, y)^t ~|~B_z (x, y)^t \leq C_z,~(x, y)\geq 0\},$$
where vector $y$ represents the time and function variables, and $x$ the burst parameters of the flows, and $(x, y)^t$ is the transposition of the line vector $(x, y)$.

This linear program has the following properties (note that they only concern coefficient that relates to the variables $x$ and not to the variables $y$).

\begin{itemize}
	\item[$(P_1)$] For all constraint $c$, for all $z'$, $(B_z)_{c, z'} \leq 0$ and $[(B_z)_{c, z'} < 0 \land (B_z)_{c, z''} < 0] \implies z'= z''$. In other words, there is at most one variable of type $x_{z'}$ in each constraint, and these variables appears as upper bounds. 
	\item[$(P_2)$] For all $z'$, $(A_z)_{z'} \geq 0$: the objective is increasing with the burst parameters $x_{z'}$.
\end{itemize}

Our aim is then to solve 
\begin{equation}
\label{eq:f1}
\sup \{x~|~x\leq \cL(x)\} = \sup \{x~|~x_z\leq  \max\{A_z (x, y)^t ~|~B_z (x, y)^t \leq C_z,(x, y) \geq 0\}\}.
\end{equation} 

We now show that this problem is equivalent to extracting the burst parameters of the following linear program: 
\begin{equation}
\label{eq:f2}
\max\{\sum_z x_z~|~x_z \leq A_z (x, y_z)^t, B_z (x, y_z)^t \leq C_z,~(x, y)\geq 0, \text{ for all } z\in Z\}.
\end{equation} 
Note that the set of variables $y_z$ are disjoint for each linear program. We call a solution of~\eqref{eq:f2} a vector $x$ for which their exists $((y_z)_{z\in Z})$ satisfying the constraints of ~\eqref{eq:f2}, and we call the solution {\em optimal} if it maximizes the sum of its coefficients. 

\begin{lemma}
	The optimal solution of~\eqref{eq:f2} is unique. 
\end{lemma}

\begin{proof}
	Suppose that $(x, (y_z)_{z\in Z})$ and $(x', (y'_z)_{z\in Z})$ are two different optimal solutions. The equality  $\sum_z x_z = \sum_z x'_z$ then holds and $x\neq x'$.
	
	 We will show that there exists another solution $(\tilde{x}, (\tilde{y}_z)_{z\in Z})$ with  $\tilde{x} = \max(x, x')$ where the maximum is coordinate-wise. In that case, $\sum_z \tilde{x}_z > \sum_z x_z$, which is in contradiction with the optimality of  $(x, y_z)$ and ends the proof. 
	
	From Property $(P_1)$, for all $c$, $(B_i(\tilde{x}, y_z))_c \leq (B_i(x, y_z))_c \leq (C_z)_c$, so the constraints are still satisfied when $x$ is replaced by $\tilde{x}$, and similarly for $x'$.  
	
	For all $z\in Z$, as $\tilde{x}_z = \max(x_z, x'_z)$, we either have  $\tilde{x}_z = x_z \leq A_z (x, y_z)$ or $\tilde{x}_z = x'_z \leq A_z (x', y'_z)$. From Property $(P_2)$, $A_z (x, y_z) \leq A_z (\tilde{x}, y_z)$, so if  $\tilde{x}_i = x_i$, one can choose $\tilde{y}_z = y_z$, and otherwise choose $\tilde{y}_z = y'_z$.  All constraints are satisfied, so $(\tilde{x},(\tilde{y}_z)_{z\in Z})$ is a solution of \eqref{eq:f2}. 
\end{proof}

\begin{theorem}
	The two following statements are equivalent.
	\begin{enumerate}
		\item $x$ is the maximal solution of~\eqref{eq:f1}.
		\item $x$ is the vector of variables extracted from the optimal solution of~\eqref{eq:f2}.
	\end{enumerate}
\end{theorem}

\begin{proof}
	If we show that any solution of~\eqref{eq:f1} is a solution of \eqref{eq:f2} and conversely, then the uniqueness of the maximal/optimal solutions to the two problems is enough to conclude.
	
	Let $x$ be a solution of~\eqref{eq:f1}, and $(x, y_z)$ be a solution of the sub-problem $z$. Then $(x, (y_z)_{z\in Z})$ is a solution to~\eqref{eq:f2}. Conversely, if $(x, (y_z)_z)$ is a solution to~\eqref{eq:f2},  $(x, (y_z))$ is a solution to sub-problem $x$ of ~\eqref{eq:f1} and $x$ is a solution of~\eqref{eq:f1}.
\end{proof}

\paragraph{Uniqueness of the fix-point} We just exhibited a linear program whose optimal solution is the largest solution fix-point of $\{x\leq \cL(x)\}$. It might be seen as pessimistic as the intuition would be that the smallest fix-point will also give an admissible solution for the cyclic network. This has been proved in the TFA++ analysis in~\cite{TBM19} and in many classical cases, the fix point is unique, because the considered system is linear. 
 
In this paragraph, we show that the solution of the fix-point with Properties $(P_1)$ and $(P_2)$ is unique. In this paragraph, we assume that the maximal solution of \eqref{eq:f1} is finite.

\begin{lemma}
	For all $z\in Z$, $\cL_z$ is concave and non-decreasing. 
\end{lemma}

\begin{proof}
	To prove the result, let us rewrite function $\cL_z(x)$ using the duality of the linear program. 
	First, we separate variables $x$ and $y$. 
	From Property $(P_1)$, one can transform $B_z(x, y)^t \leq C_z$ into $B'_z y^t \leq C'_z(x)$, where now the coefficients of $C'_i$ depend on $x$. Property $(P_1)$ tells us that the coefficients of $C'_z(x)$ are linear and non-decreasing in each variable of $x$. 
	From Property $(P_2)$, one can rewrite $A_z(x, y)^t$ as $A'_z y^t + A''_z x^t$, where all the coefficients of $A''_z$ are non-negative. 
	
	We then have $\cL_z(x) = \max\{A'_z y^t ~|~ B'_z y^t \leq C'_z(x)\} + A''_z x^t$. As function $\cL_z$ represents the computation of a maximum backlog in a stable feed-forward neworks (variables $x$ represent the burst parameters), the linear problem involved in $\cL_z(x)$ has an optimal solution. The dual problem then has the same optimal solution, and we can express $\cL_z(x)$ as 
	$$\cL_z(x) = \min\{C'_z(x)^t w^t ~|~ B_z^{'t} w^t \geq A_z^{'t}\} + A''_z x^t.$$
	The polyhedron defined by $ B_z^{'t} w^t \geq A_z^{'t}$ does not depend on $x$. The optimal solution is obtained at a vertex of this polyhedron, and there is a finite number of vertices. Then $\cL_z(x)$ is the minimum of non-decreasing linear functions (the coefficients of $x$ are all non-negative). Then $\cL_z(x)$ is non-decreasing and concave.
\end{proof} 

To prove the uniqueness of the fix-point, we will follow the lines of the proof of~\cite{Ken01}, where the result is proven for strictly concave functions and strict quasi-increasing function. These assumptions do not hold as functions $\cL_z$ are piece-wise linear. We then adapt the proof in the case where $\cL_z(0) > 0$ for concave and quasi-increasing functions.  The adaptation in straightforward, but for sake of completeness, let use write it, and then comment on the additional hypothesis $\cL_z(0) >0$.

\begin{definition}
	A function $g = (g_1, \ldots, g_n):\R_+^n \to \R_+^n$ is quasi-increasing if for all $i$, for all $x, y$, such that $y\geq x$ and $x_i = y_i$, $g_i(y) \geq g_i(x)$.
\end{definition}

In our case, as $\cL$ in non-decreasing, $\cL - Id$ is quasi-increasing, where $Id$ is the identity function.

\begin{lemma}
	\label{lem:fp}
	Let $F$ be a function from $\R_+^n$ to $\R_+^n$, that is concave and such that $F - Id$ is quasi-increasing and $F_i(0) > 0$ for all $i\leq n$. 
	If $F$ has a fix-point, then this fix-point is unique.
\end{lemma}

\begin{proof}
	Suppose $x$ and $y$ are two fix-points of $F$: $F(x) = x$ and $F(y) = y$. 
	Define $\gamma = \min_{j\leq n} (\frac{x_j}{y_j}) = \frac{x_r}{y_r}$. If $\gamma \geq 1$, then $x\geq y$. 
	Suppose now that $\gamma <1$ and define $w = \gamma y$.
	On the one hand, $w_i < y_i$ for all $i$. So by concavity of $F$, 
	we have $F_i(w) \geq (1-\gamma)F_i(0) + \gamma F_i(y) >  \gamma F_i(y)  = \gamma y_i = w_i$. In particular, $F_r(w) > w_r$. 
	
	On the other hand, $w\leq x$ and $w_r = x_r$. Indeed,  $w_j = \gamma y_j \leq \frac{x_j}{y_j} y_j = x_j$, and the inequality becomes an equality when $i=r$. 
	As $F - Id$ is quasi-increasing, then $F_r(x) - x_r \geq F_r(w) -w_r$. 
	
	But $x$ is a fix-point, so combining the obtained inequalities, we get $0 = F_r(x) - x_r \geq F_r(w) -w_r > 0$, which is a contradiction. So one must have $x \geq y$. 
	
	By inverting $x$ and $y$, we also obtain $x\leq y$, and finally $x = y$ and enables to conclude regarding the uniqueness of the fix-point.
\end{proof}

Back to our case, the uniqueness of the fix-point is granted if $\cL_z(0)>0$ for all $z\in Z$. This might not always be the case. For example, when the initial bursts and the latencies of the servers are all null.
However, in the general case, this assumption should hold. As we do not shape the flows of interest to compute the burst, we have the following properties: consider a flow with arrival curve $\gamma_{b,r}$ crossing a server with service curve $\beta_{R,T}$. The maximal backlog in the server is $b + rT$, so if either $b$ or $T$ is non-null, the backlog bound transmitted for the bound of the next flow is non-null. Then there exists $k$ such that $\cL_z^k(0)>0$ for all $z$, where $F^k$ is the $k$-th iteration of the composition of $F$.  Here, for example, $k$ can be chosen as $\max_iK_i$. 
Because $\cL$ is non-decreasing and concave, so is $\cL^k$. Then Lemma~\ref{lem:fp} can be applied to $\cL^k$. As a fix-point of $\cL$ is necessary a fix-point of $\cL^k$, this proves the uniqueness of the fix-point of $\cL$.

\subsection{Application to TFA++}
\label{sec:lp-tfa}
Before giving the linear program for the FIFO linear programming with cyclic dependencies, let us give an alternate formulation of the solution with TFA++ given in~\cite{TBM19}. In that paper, the authors compute performances in cyclic network using the TFA++ method for each server. 

Two limitations can be overcome by using the linear programming approach.
\begin{enumerate}
	\item The authors assume that they take into account the shaping only when the shaping rates exceed some given value. The reason for this seems to be the simplification of the computation of maximum delay, and more precisely the place where the horizontal distance between the aggregate arrival curve and the service curve is maximized. With the linear programming approach, this place is not directly computed but given as the solution of a linear program. It also allows more complex service curves than rate-latency. 
	\item The method chosen to compute the fix-point is iteration from 0. The authors show that the least fix-point is indeed a valid solution for computing the performance bounds, which is an improvement compared to~\cite{BBL18} that states this result for the greater fix-point. As we showed the uniqueness of the fix-point, so the two previous approaches were in fact similar. Here, the linear program avoids the iteration whose raw output would be a lower approximation of the fix-point (the iterated are lower bounds of the fix point) and thus require additional (yet simple) computations to obtain  upper bounds.    
\end{enumerate}

\paragraph{Linear program for the TFA++ method}
Table~\ref{tab:tfa-lp} gives the linear constraints for each server $j$ to compute the 
maximum horizontal distance between the arrival and the departure processes.
We use the notation $A_i^{(j)}$ (resp. $D^{(j)}_i$) for the arrival (resp. departure) process of flow $f_i$ at server $j$. As the analysis is TFA, the linear programs are local to each servers, except the delays, so we do note enforce the equality between the departure process at a server and the arrival process at its successors. The dates $s_j$ and $t_j$ are respectively the arrival and departure time of the bit of data suffering the largest delay in server $j$.

\begin{table}[htbp]	
	\centering
	\begin{tabular}{|lll|}
		\hline
		Maximize & $\sum_j \bd_j$&\\
		\hline
		\hline
		such that & for all server $j$&\\ 
		&$0\leq \bs_j \leq \bt_j$ & time constraints\\
		$\forall i \in \fl(j)$ &$\bA^{(j)}_i(\bs_{j}) \leq \bx^{(j)}_i + r_i \bs_{j}$ & arrival constraints\\
		$\forall h$& $\sum_{i\in \fl(h, j)} \bA^{(j)}_i(\bs_j) \leq L_h + C_h \bs_j$&shaping constraints\\
		&$\sum_{i\in\fl(j)}\bD^{(j)}_i(\bt_{j}) \geq R_j\bt_{j} - R_jT_j$& service constraint\\
		&$\sum_{i\in\fl(j)}\bD^{(j)}_i(\bt_{j}) \geq 0$& \\
		&$\bA^{(j)}_i(\bs_{j}) = \bD^{(j)}_i(\bt_{j})$& FIFO constraint\\
		&$\bd_j = \bt_j - \bs_j$ & delay at server $j$\\
		$\forall i \in \fl(j)$& $\bx^{(\su(j))}_i = \bx^{(j)}_i + r_i \bd_j$ & burst propagation\\
		\hline
		& for all flow $f_i$&\\
		&$\bx_i^{(\pi_i(0))} = b_i$& initial bursts\\
		\hline
	\end{tabular}
\caption{Linear program for TFA++ with cyclic dependencies: server $j$}
\label{tab:tfa-lp}
\end{table}

\section{Experimental results}
\label{sec:numerical}
In the first part of the experimental result, we use very simple topologies for the comparisons, that is tandem networks and the ring topology. In the second part, we use networks topologies met in real applications. 

The algortihm presented in this paper have been implemented in {\tt Python 3}. The linear programs methods are solved in two steps. First a linear program is generated (with a {\tt Python} script), and then solves with the open-source linear programming solver {\tt lp\_solve}. This implementation is not optimal for at least two reasons. First commercial like {\tt Cplex} and {\tt Gurobi}  and other open source can solve linear programs more than 100 times faster~\cite{LP2013}. Second, there exists solutions to use {\tt lp\_solve} directly inside Python code. We chose this way of doing this to be able to obtain readable linear programs and be able to interpret more simply the solution. 
This being said, the comparison of the computation times between non-LP solutions and LP solutions can be unfair, but the comparison between LP solutions is still valid. 

We first test our algorithms on toy examples, like tandems and rings, and then on real cases. 
\subsection{Toy examples}
We will compare 4 methods we introduced in this paper: TFA++, SFA, LP (linear program from~\cite{BS12}), PLP (polynomial-size linear program from Section~\ref{sec:tree}), and the performances obtained when the network has regulators after each server~\cite{LeB18}. Regulators regulate the flows according to their arrival curve. It is shown that the delays induces by them is not modified, but the end-to-end delays are computed in a way similar to the TFA algorithm recalled in Algorithm~\ref{algo:tfa}, except that in line 5, $b_i^{(succ_i(j))}$ remains $b_i^{(j)}$. When we mentionned the delay bounds with regulators, it is the bounds computed this way, and not the worst-case delay bounds when inserting regulators.

\subsubsection{Tandem networks}
A tandem network with $n$ servers is a network whose underlying graph is a line. By convention, we number the servers from 1 to $n$ in the topological order. 
We assume uniform networks: each server has offers a service curve $\beta:R(t-T)_+$, a maximum service curve $\beta_u:t\mapsto \eta Rt$, with $\eta \geq 1$ and each flow is constrained by the arrival curve $\alpha:t\mapsto b + rt$. 

We fix $T = 0.001s$, $b = 1Kb$ and $R = 10$Mbps.  The arrival $r$ will vary to study the network at different load, and $\eta$ vary in order to study the sentitivity of TFA++ and PLP to the maximum service rate. 

The flow of interest (f.o.i.) crosses all servers and we will study 
two different configurations: two-hop cross-traffic, and source-sink networks. 

\paragraph{Two-hop cross-traffic}
There are $n-1$ interfering flows, with path $\langle i,i + 1 \rangle$ for all $1\leq i < n$ as depicted on Figure~\ref{fig:2hop} for 4 servers.

\begin{figure}[hpbp]
	\centering
	\begin{tikzpicture}[server/.style={shape=rectangle,draw,minimum height=.8cm,inner xsep=3ex}]
	\node[server,name=S1] at (0,0) {$1$};
	\node[server,name=S2] at (2,0) {$2$};
	\node[server,name=S3] at (4,0) {$3$};
	\node[server,name=S4] at (6,0) {$4$};
	\draw[->, thick, red] (-1,0) node[left] {$f_1$} -- (7,0);
	\draw[->, thick, blue] (-1,0.3) node[left] {$f_2$} to[out=-10,in=-170] (2.7,.3);
	\draw[->, thick, green!70!black]  (1,-0.3) node[left] {$f_3$} to[out=10,in=170] (5,-.3);
	\draw[->, thick, purple] (3.3,0.3) node[left] {$f_4$} to[out=-10,in=-170] (7,.3);
	\end{tikzpicture}
	\caption{Two-hop tandem network.}
	\label{fig:2hop}
\end{figure}
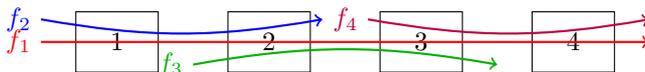


If the arrival rate is $r$, than the load of the network is $U = 3r/R$.

Figure~\ref{fig:two-hop-n} shows the delay bound obtained when the number of servers varies from 1 to 25, when $\eta = 1$ and $U = 0.5$. The LP method is computed only up to 7 servers. 
One can check that the linear programming give the tightest bound, followed by the PLP bound, which confirms to our results. 

The regulators also allow to obtain slightly better bounds for tandems longer than 15, which seems also intuitive as the bursts cannot propagate in the network, and regulating the flows has then more effect on longer networks. 
More surprisingly, one can see that the SFA method, that does not take into account the shaping effect of the maximal service curve behaves almost like TFA++, specially for long network. For the tandem of length 25, the gain between TFA++ and PLP is 28\% and between SFA and PLP 29\%.

Figure~\ref{fig:two-hop-n-exec} compares the execution time of the two linear programming methods. One can verify that the PLP method, whose  execution time is below 5 seconds for 25 servers, scales much better than the LP method, whose execution time is already 10 seconds for 7 servers. 

\begin{figure}[htbp]
	\centering
	\subfloat[\label{fig:two-hop-n}Delay]{
		\centering
		\begin{tikzpicture}
			\begin{axis}[height = 5cm, width = 5cm,
				xlabel={Number of servers},
				ylabel = {Delay of the f.o.i} (s),
				legend entries={TFA++, SFA, PLP, LP, Regulators}, legend style = {fill=none, at = {(0,1)}, anchor = north west, 
					draw = none, font=\tiny},
				xmin = 1, xmax = 25,ymin=0, ymax = 0.05, axis x line = bottom, axis y line = left,
				unbounded coords = jump]
				\addplot+[mark = none, thick, cyan] table[x index = 0,y index= 1] {two_hop_n.data};
				\addplot+[mark = none, thick, lime!70!black] table[x index = 0,y index= 2] {two_hop_n.data};
				\addplot+[mark = none, thick, red] table[x index = 0,y index= 3] {two_hop_n.data};
				\addplot+[mark = none, ultra thick, violet] table[x index = 0,y index= 5] {two_hop_n.data};
				\addplot+[mark = none, thick, dashed, gray] table[x index = 0,y index= 4] {two_hop_n.data};		
			\end{axis}
		\end{tikzpicture}}
	\hspace{0.4cm}
	\subfloat[\label{fig:two-hop-n-exec}Execution time]{
		\centering
		\begin{tikzpicture}
			\begin{semilogyaxis}[height = 5cm, width = 5cm,
				xlabel={Number of servers},
				ylabel = {Execution time (s)},
				legend entries={PLP, LP}, legend style = {fill=none, at = {(0.0,0.97)}, anchor = north west, draw = none, font=\tiny},
				xmin = 1, xmax = 25,ymin=0, ymax = 10, axis x line = bottom, axis y line = left,
				unbounded coords = jump]
				\addplot+[mark = none, thick, red] table[x index = 0,y index= 3] {two_hop_n_exec.data};
				\addplot+[mark = none, thick, violet] table[x index = 0,y index= 5] {two_hop_n_exec.data};
			\end{semilogyaxis}
		\end{tikzpicture}}
	\caption{Comparisons of different methods for two-hop cross-traffic when varying the number of servers.}
\end{figure}
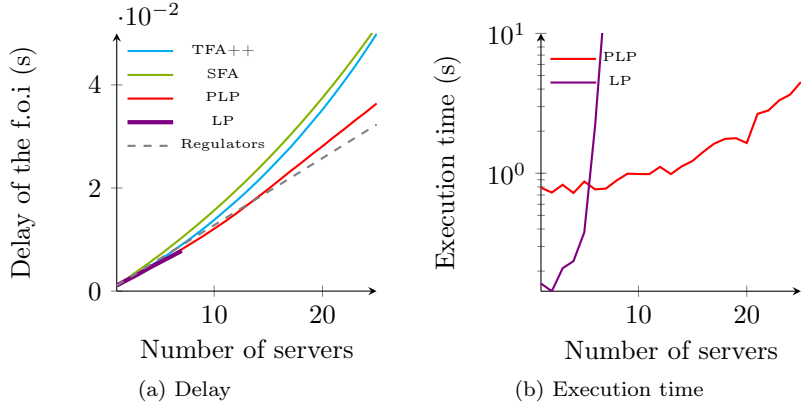

Figure ~\ref{fig:two-hop-load} shows the delay bounds obtained for a network of 10 servers in function of the load. We discard the LP method.
One can observe that TFA++ is outperformd by SFA for high loads. Indeed, when the arrival rate is exactly equal to the service rate in a FIFO server, there can be no benefit from the shaping. One can also see that when $U>0.7$, the delay computed with PLP method grows faster. This is due to the fact that the TFA++ constraints do not enable to improve the delay bounds anymore (they are two pessimistic).
 Finally, Figure~\ref{fig:two-hop-K} show the sensitivity of the delay to the maximum service rate. We here compare the servers for maximum service rates $\eta R$, for $\eta\in\{1, 1.2, 2\}$.
Already when $\eta=1.2$.  TFA++ is outperformed by SFA (for a load above 0.3). While TFA++ seems very sensitive to the maximum serice rate, the delays computed with PLP do not vary much. 
 This shows the effectiveness of PLP approach  when the maximum service rate does not equal the service rate. 

 \begin{figure}[htbp]
 	\centering
 	\subfloat[\label{fig:two-hop-load}Delay in function of the load]{
 		\centering
 		\begin{tikzpicture}
 			\begin{semilogyaxis}[height = 5cm, width = 5cm,
 				xlabel={Load of the network},
 				ylabel = {Delay of the f.o.i (s)},
 				legend entries={TFA++, SFA, PLP, Regulators}, legend style = {fill=none, at = {(0,1)}, anchor = north west, draw = none, font=\tiny},
 				xmin = 0, xmax = 1,ymin=0, ymax = 0.1, axis x line = bottom, axis y line = left,
 				unbounded coords = jump]
 				\addplot+[mark = none, thick, cyan] table[x index = 0,y index= 1] {two_hop_K1.data};
 				\addplot+[mark = none, thick, lime!70!black] table[x index = 0,y index= 2] {two_hop_K1.data};
 				\addplot+[mark = none, thick, red] table[x index = 0,y index= 3] {two_hop_K1.data};
 				\addplot+[mark = none, dashed, thick, gray] table[x index = 0,y index= 4] {two_hop_K1.data};		
 			\end{semilogyaxis}
 		\end{tikzpicture}}
 	\hspace{0.4cm}
 	\subfloat[\label{fig:two-hop-K}Shaping effect]{
 		\centering
 			\begin{tikzpicture}
 			\begin{semilogyaxis}[height = 5cm, width = 5cm,
 				xlabel={Load},
 				ylabel = {Delay of the f.o.i. (s)},
 				legend entries={SFA, 
 					TFA++ $\eta = 1$, 
 					PLP $\eta = 1$, 
 					TFA++ $\eta = 1.2$, 
 					PLP $\eta = 1.2$, 
 					TFA++ $\eta = 2$, 
 					PLP $\eta = 2$}, 
 				legend style = {fill=none, at = {(0.03,0.97)}, anchor = north west, draw = none, font=\tiny},
 				xmin = 0, xmax = 1,ymin=0, ymax = 0.1, axis x line = bottom, axis y line = left,
 				unbounded coords = jump
 				]
 				\addplot+[mark = none, thick, lime!70!black] table[x index = 0,y index= 2] {two_hop_K1.data};
 				\addplot+[mark = none, thick, dotted, cyan] table[x index = 0,y index= 1] {two_hop_K1.data};
 				\addplot+[mark = none, thick, dotted, red] table[x index = 0,y index= 3] {two_hop_K1.data};
 				\addplot+[mark = none, thick, dashed, cyan] table[x index = 0,y index= 1] {two_hop_K10.data};
 				\addplot+[mark = none, thick, dashed, red] table[x index = 0,y index= 3] {two_hop_K10.data};
 				\addplot+[mark = none, thick, solid, cyan] table[x index = 0,y index= 1] {two_hop_K2.data};
 				\addplot+[mark = none, thick, solid, red] table[x index = 0,y index= 3] {two_hop_K2.data};	
 			\end{semilogyaxis}
 		\end{tikzpicture}}
 	\caption{Comparisons of different method for two-hop cross-traffic when varying the load.}
 \end{figure}
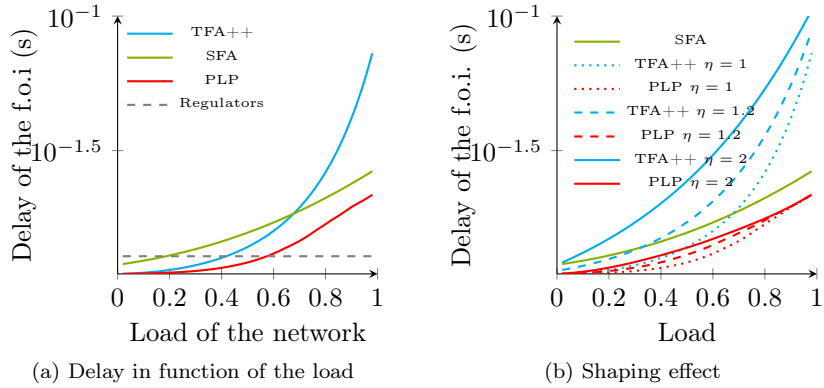

\paragraph{Source-sink network}
We call source-sink tandem a tandem with $n$ servers and $2n-1$ flows. Each flows either starts at server 1 or send at server $n$. In the uniform case, there is one flow per possible path, as depicted in Figure~\ref{fig:long} with $n = 4$. 
There are $n$ flows crossing each server, so the load of the network is $nr/R$. 
 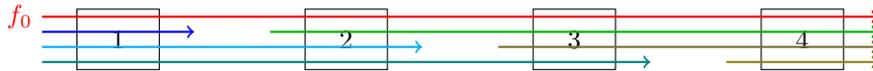
\begin{figure}[htbp]
	\centering
	\begin{tikzpicture}[server/.style={shape=rectangle,draw,minimum height=.8cm,inner xsep=3ex}]
	\node[server,name=S1] at (0,0) {$1$};
	\node[server,name=S2] at (3,0) {$2$};
	\node[server,name=S2] at (6,0) {$3$};
	\node[server,name=S2] at (9,0) {$4$};
	\draw[->, thick, red] (-1,0.3) node[left] {$f_0$} -- (10, 0.3);
	\draw[->, thick, blue] (-1,0.1) -- (1,0.1);
	\draw[->, thick, cyan] (-1,-0.1) -- (4,-0.1);
	\draw[->, thick, teal] (-1,-0.3) -- (7,-0.3);
	\draw[->, thick, green!70!black] (2,0.1) -- (10,0.1);
	\draw[->, thick, olive!70!black] (5,-0.1) -- (10,-0.1);
	\draw[->, thick, olive!90!black] (8,-0.3) -- (10,-0.3);
	\end{tikzpicture}
	\caption{Source-sink tandem of length 4.}
	\label{fig:long}
\end{figure}

Figure~\ref{fig:long-n} depicts the worst-case delays computed by each method when the length of the tandem grows from 1 to 25 when $\eta = 1$ and $U=0.5$. One can still check that thee linear programming methods still give the best delay bounds. Here, the gap between the LP and PLP methods is very small. The TFA++ also performs very well, the gain for 25 severs is 13\%. These three bounds are below the bound obtained with regulators. SFA is completely outperformed. Indeed, at each server, $n-1$ flows continue to the next server. Then the shaping has a very strong effect on the performances. 
Figure~\ref{fig:long-n-exec} compares the execution times of LP and PLP. Again we see the tractability improvement of this new approach. 

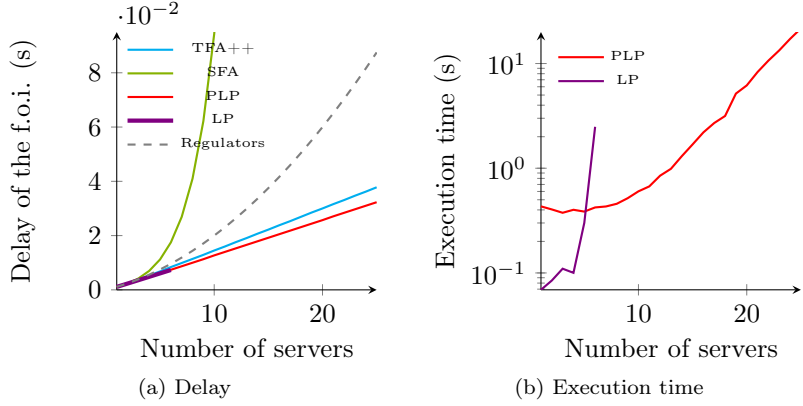
\begin{figure}[htbp]
	\centering
	\subfloat[\label{fig:long-n}Delay]{
		\centering
		\begin{tikzpicture}
			\begin{axis}[height = 5cm, width = 5cm,
				xlabel={Number of servers},
				ylabel = {Delay of the f.o.i. (s)},
				legend entries={TFA++, SFA, PLP, LP, Regulators}, legend style = {fill=none, at = {(0,1)}, anchor = north west, draw = none, font=\tiny},
				xmin = 1, xmax = 25,ymin=0, ymax = 0.095, axis x line = bottom, axis y line = left,
				unbounded coords = jump]
				\addplot+[mark = none, thick, cyan] table[x index = 0,y index= 1] {long_n.data};
				\addplot+[mark = none, thick, lime!70!black] table[x index = 0,y index= 2] {long_n.data};
				\addplot+[mark = none, thick, red] table[x index = 0,y index= 3] {long_n.data};
				\addplot+[mark = none, ultra thick, violet] table[x index = 0,y index= 5] {long_n.data};
				\addplot+[mark = none, thick, dashed, gray] table[x index = 0,y index= 4] {long_n.data};		
			\end{axis}
		\end{tikzpicture}}
	\hspace{0.4cm}
	\subfloat[\label{fig:long-n-exec}Execution time]{
		\centering
		\begin{tikzpicture}
			\begin{semilogyaxis}[height = 5cm, width = 5cm,
				xlabel={Number of servers},
				ylabel = {Execution time (s)},
				legend entries={PLP, LP}, legend style = {fill=none, at = {(0.03,0.97)}, anchor = north west, draw = none, font=\tiny},
				xmin = 1, xmax = 25,ymin=0, ymax = 20, axis x line = bottom, axis y line = left,
				unbounded coords = jump]
				\addplot+[mark = none, thick, red] table[x index = 0,y index= 3] {long_n_exec.data};
				\addplot+[mark = none, thick, violet] table[x index = 0,y index= 5] {long_n_exec.data};
			\end{semilogyaxis}
		\end{tikzpicture}}
	\caption{Comparisons of different method for source-sink networks when varying the number of servers.}
\end{figure}

Figure~\ref{fig:long-load} show how the delay bounds grow with the load of a tandem of length 10. When the load is small, TFA++ and PLP are similar, but the gap between the two grows exponentially. For example, the gain between TFA+ and PLP is 12\% improvement for a load of 0.5 and 51\% for a load of 0.8. PLP outperforms the delay bound with regulators until a load of 0.95. 

Figure~\ref{fig:long-K} compares the delay bounds for several maximum service rates, with $\eta\in\{1, 2, 3\}$. Again, the PLP method is not very sensitive to this parameter, and the bounds for $\eta=2$ and $\eta=3$ are very similar, contrary to the TFA++ method. TFA++ and SFA are comparable when $\eta=3$.   

 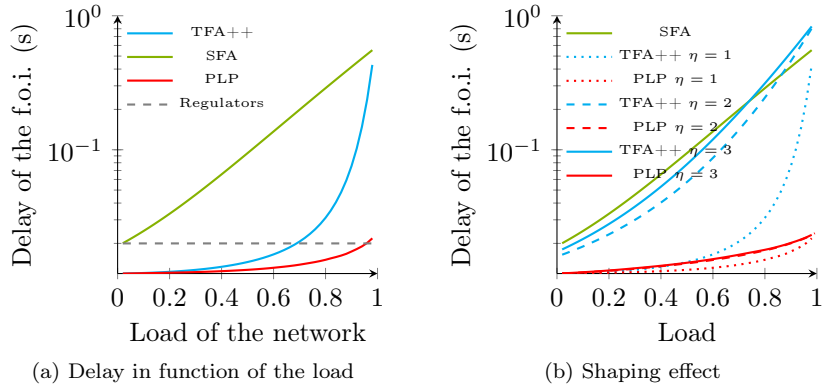
\begin{figure}[htbp]
	\centering
	\subfloat[\label{fig:long-load}Delay in function of the load]{
		\centering
		\begin{tikzpicture}
			\begin{semilogyaxis}[height = 5cm, width = 5cm,
				xlabel={Load of the network},
				ylabel = {Delay of the f.o.i.} (s),
				legend entries={TFA++, SFA, PLP, Regulators}, legend style = {fill=none, at = {(0,1)}, anchor = north west, draw = none, font=\tiny},
				xmin = 0, xmax = 1,ymin=0, ymax = 1, axis x line = bottom, axis y line = left,
				unbounded coords = jump]
				\addplot+[mark = none, thick, cyan] table[x index = 0,y index= 1] {long_load.data};
				\addplot+[mark = none, thick, lime!70!black] table[x index = 0,y index= 2] {long_load.data};
				\addplot+[mark = none, thick, red] table[x index = 0,y index= 3] {long_load.data};
				\addplot+[mark = none, thick, dashed, gray] table[x index = 0,y index= 4] {long_load.data};		
			\end{semilogyaxis}
		\end{tikzpicture}}
	\hspace{0.4cm}
	\subfloat[\label{fig:long-K}Shaping effect]{
		\centering
		\begin{tikzpicture}
			\begin{semilogyaxis}[height = 5cm, width = 5cm,
				xlabel={Load},
				ylabel = {Delay of the f.o.i. (s)},
				legend entries={SFA, 
					TFA++ $\eta = 1$, 
					PLP $\eta = 1$, 
					TFA++ $\eta = 2$, 
					PLP $\eta = 2$, 
					TFA++ $\eta = 3$, 
					PLP $\eta = 3$}, 
				legend style = {fill=none, at = {(0,1)}, anchor = north west, draw = none, font=\tiny},
				xmin = 0, xmax = 1,
				ymax = 1, axis x line = bottom, axis y line = left,
				unbounded coords = jump
				]
				\addplot+[mark = none, thick, lime!70!black] table[x index = 0,y index= 2] {long_load.data};
				\addplot+[mark = none, thick, dotted, cyan] table[x index = 0,y index= 1] {long_load.data};
				\addplot+[mark = none, thick, dotted, red] table[x index = 0,y index= 3] {long_load.data};
				\addplot+[mark = none, thick, dashed, cyan] table[x index = 0,y index= 1] {long_K2.data};
				\addplot+[mark = none, thick, dashed, red] table[x index = 0,y index= 3] {long_K2.data};
				\addplot+[mark = none, thick, solid, cyan] table[x index = 0,y index= 1] {long_K10.data};
				\addplot+[mark = none, thick, solid, red] table[x index = 0,y index= 3] {long_K10.data};
			\end{semilogyaxis}
		\end{tikzpicture}}
	\caption{Comparisons of different method for source-sink networks when varying the load.}
\end{figure}

\subsubsection{Mesh network}
We consider the mesh network of Figure~\ref{fig:mesh}. There is one flow per path from server 0 or 1 to server 8, which represents a total of 16 paths. Servers 0 to 7 have the same characteristics as above, and server 8's service rate is $2R$, as there are twice as much flows crossing it compared to the other servers. We also keep the same characteristics as above for the flows. 

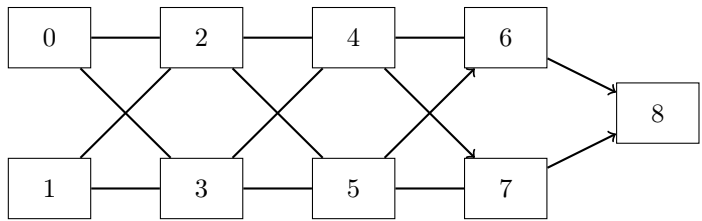
\begin{figure}[htbp]
\centering
\begin{tikzpicture}
	[server/.style={shape=rectangle,draw,minimum height=.8cm,inner xsep=3ex}]
	\node[server,name=S0] at (0,1) {0};
	\node[server,name=S1] at (0,-1) {1};
	\node[server,name=S2] at (2,1) {2};
	\node[server,name=S3] at (2,-1) {3};	
	\node[server,name=S4] at (4,1) {4};
	\node[server,name=S5] at (4,-1) {5};
	\node[server,name=S6] at (6,1) {6};
	\node[server,name=S7] at (6,-1) {7};
	\node[server,name=S8] at (8,0) {8};
	\draw[thick, ->] (S0) -- (S2) -- (S4) -- (S6) -- (S8);
	\draw[thick, ->] (S0) -- (S3) -- (S4) -- (S7);
	\draw[thick, ->] (S1) -- (S3) -- (S5) -- (S7) -- (S8);
	\draw[thick, ->] (S1) -- (S2) -- (S5) -- (S6);
\end{tikzpicture}	
\caption{Mesh network.}
\label{fig:mesh}
\end{figure}

Figures~\ref{fig:num_mesh} compares the delays obtained for TFA++, SFA and the two different methods introduced for analyzing feed-forward networks: network unfolding and flow splitting. Due to its computation time, we do not compare with the exponential LP method. Figure~\ref{fig:mesh1} depicts the delays when $\eta =1$ and Figure~\ref{fig:mesh2} when $\eta = 5$. 
Similarly to the previous cases, TFA++ is very accurate when $\eta = 1$ and the load is small, but becomes pessimistic when $\eta$ is larger or when the load is large. It is not a surprise that the unfold method leads to tighter delay bounds than the split method. Indeed, splitting a flow lead to some over-approximations. But the unfold network's size being exponential in the size of the original network, this method is not scalable. We notice that the gap between the two methods is not very large, specially when $\eta=1$. 

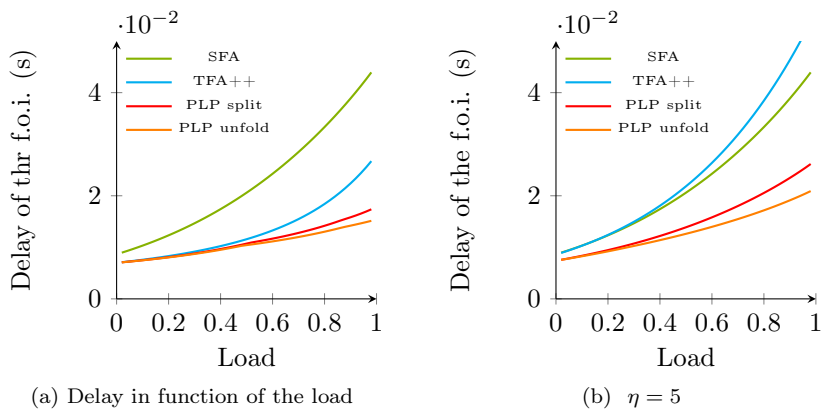
\begin{figure}[htbp]
	\centering
	\subfloat[\label{fig:mesh1}Delay in function of the load]{
		\centering
		\begin{tikzpicture}
			\begin{axis}[height = 5cm, width = 5cm,
				xlabel={Load},
				ylabel = {Delay of thr f.o.i. (s)},
				legend entries={SFA, TFA++, PLP split, PLP unfold}, 
				legend style = {fill=none, at = {(0,1)}, anchor = north west, draw = none, font=\tiny},
				xmin = 0, xmax = 1,ymin=0, ymax = 0.05, axis x line = bottom, axis y line = left,
				unbounded coords = jump
				]
				\addplot+[mark = none, thick, lime!70!black] table[x index = 0,y index= 2] {mesh1.data};
				\addplot+[mark = none, thick, cyan] table[x index = 0,y index= 1] {mesh1.data};
				\addplot+[mark = none, thick, red] table[x index = 0,y index= 3] {mesh1.data};
				\addplot+[mark = none, thick, orange] table[x index = 0,y index= 4] {mesh1.data};
			\end{axis}
		\end{tikzpicture}}
	\hspace{0.4cm}
	\subfloat[\label{fig:mesh2} $\eta = 5$]{
		\centering
		\begin{tikzpicture}
			\begin{axis}[height = 5cm, width = 5cm,
				xlabel={Load},
				ylabel = {Delay of the f.o.i. (s)},
				legend entries={SFA, TFA++, PLP split, PLP unfold}, 
				legend style = {fill=none, at = {(0,1)}, anchor = north west, draw = none, font=\tiny},
				xmin = 0, xmax = 1,ymin=0, ymax = 0.05, axis x line = bottom, axis y line = left,
				unbounded coords = jump
				]
				\addplot+[mark = none, thick, lime!70!black] table[x index = 0,y index= 2] {mesh2.data};
				\addplot+[mark = none, thick, cyan] table[x index = 0,y index= 1] {mesh2.data};
				\addplot+[mark = none, thick, red] table[x index = 0,y index= 3] {mesh2.data};
				\addplot+[mark = none, thick, orange] table[x index = 0,y index= 4] {mesh2.data};
			\end{axis}
		\end{tikzpicture}}
	\caption{Comparisons of different method for a mesh network.}
	\label{fig:num_mesh}
\end{figure}

\subsubsection{Ring network}
We now consider a ring network, such as depicted on Figure~\ref{fig:ring} for $n=4$. For a ring of length $n$, there are $n$ flows of length $n$. 
If the arrival rate is $r$, than the load of the network is $U = nr/R$.
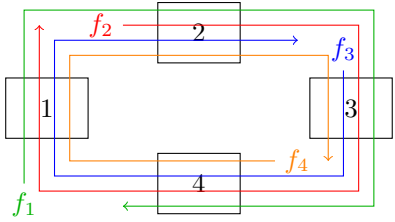
\begin{figure}[htbp]
	\centering
	\begin{tikzpicture}[server/.style={shape=rectangle,draw,minimum height=.8cm,inner xsep=3ex}]
	\node[server,name=S1] at (-2,0) {$1$};
	\node[server,name=S2] at (0,1) {$2$};
	\node[server,name=S3] at (2,0) {$3$};
	\node[server,name=S3] at (0,-1) {$4$};
	\draw[->, green!70!black] (-2.3,-1) node[below] {$f_1$} -- (-2.3,1.3)  -- (2.3, 1.3) -- (2.3, -1.3) -- (-1, -1.3);
	\draw[->, red] (-1,1.1) node[left] {$f_2$} -- (2.1,1.1) -- (2.1, -1.1) -- (-2.1, -1.1) -- (-2.1, 1.1);
	\draw[->, blue] (1.9,0.5) node[above] {$f_3$} -- (1.9,-0.9) -- (-1.9, -0.9) -- (-1.9, 0.9) -- (1.3, 0.9);
	\draw[->, orange] (1,-0.7) node[right] {$f_4$} -- (-1.7,-0.7) -- (-1.7, 0.7) -- (1.7, 0.7) -- (1.7, -0.7);
	\end{tikzpicture}
	\caption{Ring network with $n = 4$.}
	\label{fig:ring}
\end{figure}
Figure~\ref{fig:ring-n} shows the worst-case delay bound for the different methods when the number of servers grows from 2 to 10 (and to 4 for the LP method). The load of the network is 0.5 and $\eta = 1$. One can see that the bounds found for PLP and TFA++ are really close from one another and the gap with LP is larger than with the other topologies. Here again, SFA gives very inaccurate bounds. This in inline with the example of the source-sink tandem: at each server $n-1$ flows are shaped together, which makes the TFA++ method very efficient. The execution time of LP and PLP is depicted in Figure~\ref{fig:ring-n-exec}. One can see that PLP takes longer to compute, because the fix-point requires to solve a linear program a much larger linear program. 

Figure~\ref{fig:ring-load} compares the different approches (except LP) when the load of the network grows from 0 to 1, for a ring network of 7 servers. Similarly to the other examples, when the load becomes large, the PLP method computes much tighter bounds than TFA++, and has a larger stability region (local stability). The influence of TFA++ on PLP is also more visible: when the TFA++ delay bounds become infinite, the delay bounds of PLP increases more. Again, Figure~\ref{fig:ring-K} shows how the performances evolve when the maximum service rate of the servers grows, for $\eta \in \{1, 2, 5\}$. 
When $\eta = 5$, the delays of TFA++ are comparable with SFA. With TFA++, the stability region also decreases with $\eta$: the sufficient conditions for the stability computed by TFA++ for $\eta = 1, 2, 5$ are  respectively $U<0.85$, $U<0.55$ and $U< 0.38$, whereas PLP seems to ensure stability under the local stability hypothesis in all cases.

\begin{figure}[htbp]
	\centering
	\subfloat[\label{fig:ring-n}Delay]{
		\centering
		\begin{tikzpicture}
			\begin{axis}[height = 5cm, width = 5cm,
				xlabel={Number of servers},
				ylabel = {Delay of the f.o.i. (s)},
				legend entries={TFA++, SFA, PLP, LP, Regulators}, legend style = {fill=none, at = {(1,0)}, anchor = south east, draw = none, font=\tiny},
				xmin = 1, xmax = 10,ymin=0, ymax = 0.02, axis x line = bottom, axis y line = left,
				unbounded coords = jump]
				\addplot+[mark = none, thick, cyan] table[x index = 0,y index= 1] {ring_n.data};
				\addplot+[mark = none, thick, lime!70!black] table[x index = 0,y index= 2] {ring_n.data};
				\addplot+[mark = none, thick, red] table[x index = 0,y index= 3] {ring_n.data};
				\addplot+[mark = none, ultra thick, violet] table[x index = 0,y index= 5] {ring_n.data};
				\addplot+[mark = none, thick, dashed, gray] table[x index = 0,y index= 4] {ring_n.data};		
			\end{axis}
		\end{tikzpicture}}
	\hspace{0.4cm}
	\subfloat[\label{fig:ring-n-exec}Execution time]{
		\centering
		\begin{tikzpicture}
			\begin{semilogyaxis}[height = 5cm, width = 5cm,
				xlabel={Number of servers},
				ylabel = {Execution time (s)},
				legend entries={PLP, LP}, legend style = {fill=none, at = {(0.03,0.97)}, anchor = north west, draw = none, font=\tiny},
				xmin = 1, xmax = 10,ymin=0.1, ymax = 10, axis x line = bottom, axis y line = left,
				unbounded coords = jump]
				\addplot+[mark = none, thick, red] table[x index = 0,y index= 3] {ring_n_exec.data};
				\addplot+[mark = none, thick, violet] table[x index = 0,y index= 5] {ring_n_exec.data};
			\end{semilogyaxis}
		\end{tikzpicture}}
	\caption{Comparisons of different method for the ring network when varying the number of servers.}
\end{figure}
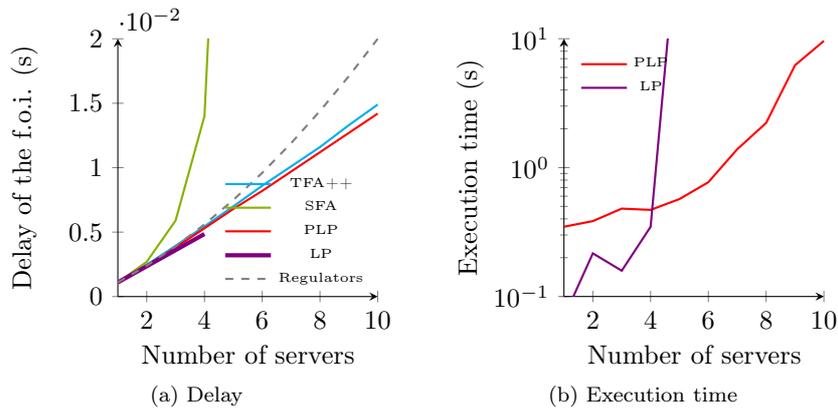

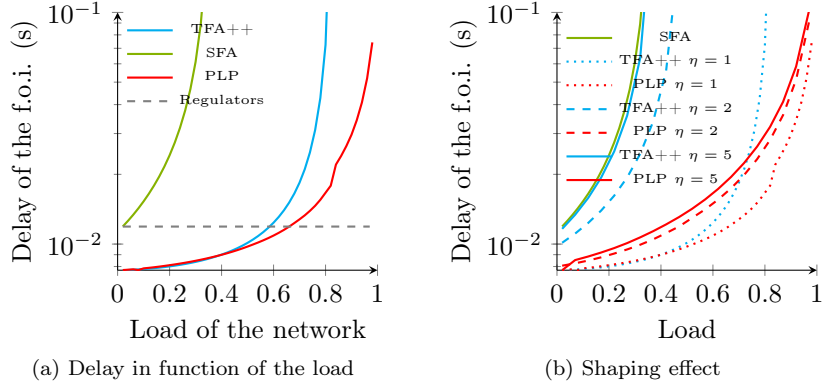
\begin{figure}[htbp]
	\centering
	\subfloat[\label{fig:ring-load}Delay in function of the load]{
		\centering
		\begin{tikzpicture}
			\begin{semilogyaxis}[height = 5cm, width = 5cm,
				xlabel={Load of the network},
				ylabel = {Delay of the f.o.i. (s)},
				legend entries={TFA++, SFA, PLP, Regulators}, legend style = {fill=none, at = {(0,1)}, anchor = north west, draw = none, font=\tiny},
				xmin = 0, xmax = 1,ymin=0, ymax = 0.1, axis x line = bottom, axis y line = left,
				unbounded coords = jump]
				\addplot+[mark = none, thick, cyan] table[x index = 0,y index= 1] {ring_load.data};
				\addplot+[mark = none, thick, lime!70!black] table[x index = 0,y index= 2] {ring_load.data};
				\addplot+[mark = none, thick, red] table[x index = 0,y index= 3] {ring_load.data};
				\addplot+[mark = none, thick, dashed, gray] table[x index = 0,y index= 4] {ring_K2.data};		
			\end{semilogyaxis}
		\end{tikzpicture}}
	\hspace{0.4cm}
	\subfloat[\label{fig:ring-K}Shaping effect]{
		\centering
			\begin{tikzpicture}
			\begin{semilogyaxis}[height = 5cm, width = 5cm,
				xlabel={Load},
				ylabel = {Delay of the f.o.i. (s)},
				legend entries={SFA, TFA++ $\eta = 1$, PLP $\eta = 1$, TFA++ $\eta = 2$, PLP $\eta = 2$, TFA++ $\eta = 5$, PLP $\eta = 5$}, 
				legend style = {fill=none, at = {(0,0.97)}, anchor = north west, draw = none, font=\tiny},
				xmin = 0, xmax = 1,ymin=0, ymax = 0.1, axis x line = bottom, axis y line = left,
				unbounded coords = jump
				]
				\addplot+[mark = none, thick, lime!70!black] table[x index = 0,y index= 2] {ring_load.data};
				\addplot+[mark = none, thick, cyan, dotted] table[x index = 0,y index= 1] {ring_load.data};
				\addplot+[mark = none, thick, red, dotted] table[x index = 0,y index= 3] {ring_load.data};
				\addplot+[mark = none, thick, cyan, dashed] table[x index = 0,y index= 1] {ring_K2.data};
				\addplot+[mark = none, thick, red, dashed] table[x index = 0,y index= 3] {ring_K2.data};
				\addplot+[mark = none, thick, cyan, solid] table[x index = 0,y index= 1] {ring_K5.data};
				\addplot+[mark = none, thick, red, solid] table[x index = 0,y index= 3] {ring_K5.data};
			\end{semilogyaxis}
		\end{tikzpicture}}
	\caption{Comparisons of different method for the ring network when varying the load.}
\end{figure}

\subsection{Real examples}
\subsubsection{Carrier networks}
Consider the network of Figure~\ref{fig:carrier}. It is made of two bidirectional rings. There are 8 flows departing from each router, 4 of them going to each of the neighbors (except the two central nodes that are not considered as neighbors). Two flows by a direct path (path of length 1) and one by a path along one ring and and the last path along the two rings, as depicted on Figure~\ref{fig:carrier}. 
Links are bidirectional, so Figure~\ref{fig:carrier} is not the exact representation of the network, that we do not give for the sake of readability.
In the example, all packets have length $L = 128 B$ and are periodically sent for each flow at period $P = 125 \mu s$. Then each flow is constrained by the arrival curve $\alpha:t\to B + t T/P$. The service rate guaranteed each link is $\beta: R(t-L/R)$. The maximum service curve, to take into account the packetization at each router is $\beta_u: t\mapsto R^u t + L$. 
The problem is to find the rate $R$ to allocate to these flows so that a maximum delay bound is satisfied for all flows. Let us assume that the target delay is $75\mu s$. 

\begin{figure}
	\centering
	\begin{tikzpicture}[scale=0.8, every node/.style={transform shape}, server/.style={shape=rectangle,draw,minimum height=.8cm,inner xsep=3ex}]
	\node[server,name=S1] at (0,1) {S};
	\node[server,name=S2] at (0,-1) {S};
	\node[server,name=S3] at (-3,1.5) {S};
	\node[server,name=S4] at (-3,-1.5) {S};
	\node[server,name=S5] at (3,1.5) {S};
	\node[server,name=S6] at (3,-1.5) {S};
	\node[server,name=S7] at (6,0) {S};
	\node[server,name=S8] at (-6,0) {S};
	\draw (S1) -- (S3) -- (S8) -- (S4) -- (S2) -- (S1) -- (S5) -- (S7) -- (S6) -- (S2); 
	\draw[->, red, thick] (-7, 0.3) -- (-6, 0.3) -- (-3.2, 1.8) -- (-3.2, 2.3);
	\draw[->, blue, thick] (-7, 0) -- (-5.8, 0) -- (-3, -1.3) -- (-0.2, -0.8) -- (-0.2, 0.8) -- (-3, 1.3) -- (-3, 2.3) ;
	\draw[green!70!black, ->, thick] (-7, -0.3) -- (-6, -0.3) -- (-3, -1.8) -- (0, -1.2) -- (3, -1.8) -- (6.3, 0) -- (3, 1.8) -- (0, 1.2) -- (-2.7, 1.8) -- (-2.7, 2.3);
	\end{tikzpicture}
	\caption{Carrier network with three types of paths: direct paths (red), one-ring path (blue), two-ring path(green).}
	\label{fig:carrier}
\end{figure}
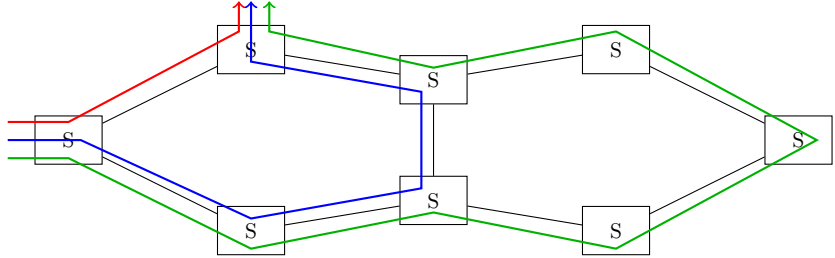

Table~\ref{tab:carrier} shows the rates to be allocated in two scenarios: first when {\em hard slicing} is at work, that is $R^u = R$, and the other when {\em soft slicing} is at work, that is $R^u = 10Gb/s$ is the total capacity of the links. 

\begin{table}[htbp]
	\centering
	\begin{tabular}{|l|l|l|l|l|}
		\hline
		& TFA++& PLP& LP\\
		\hline
		$R = R^u$ &690Mb/s (0.166s)&650Mb/s (112s)&510Mbps  (24min)\\
		$R$, $R^u = 10Gb/s$ &1.350Gb/s (0.2s) &1.0Gb/s (130s)&540Mbps (10min)\\
		\hline 
	\end{tabular}
\caption{Service rate required to guarantee the maximum delay of 75$\mu$s, with different methods.}
\label{tab:carrier}
\end{table}

One can observe that in the case of hard-slicing, the improvement of LPL cmpared to TFA++ is very small (less than 6\%). This can be explained because to obtain a delay bound as small as 75$\mu s$, the load of the network is small (approximately 15\%). But the gain obtained with LP is 26\%. In this scenario, a 24 minute computation might not be considered too costly given the gain on the bandwidth. In the case of soft slicing, PLP improves the TFA++ bound by more than 25\%, and between TFA++ and LP the improvement is 60\%. 

\subsubsection{Smart-Campus network}
In this example, the network topology is a sink tree. Four classes of flows are circulating in the network, from the leaves to the root, as depicted in Figure~\ref{fig:campus}, every class of traffic has flows following all paths, and the service policy is FIFO per class. Among the class the DRR scheduling is at work, and each flow is offered the same guarantee, that is 25\% of the service rate, and the quantum assigned to each class is $Q$.  If a network element has service curve $\beta = \beta_{R, 0}$, we assume in a simplified model that the service curve $\beta^{DRR} = \beta_{R/4, 3Q/R}$ is offered to each class. The service rate of each server is given (in Gbps) on Figure~\ref{fig:campus} and we take $Q = 16kb$. The characteristics of each class of flow is given in Table~\ref{tab:campus-desc}. Moreover, there is a shaping for each class of flow (separately) at the entrance of the network, at rate 1Gbps. The shaper is then $t\mapsto 10^9 t + \ell$, where $\ell$ is the maximum packet size of the flow of interest.

\begin{table}[htbp]
	\centering
	\begin{tabular}{|l|l|l|l|}
		\hline
		Class & burst & arrival rate & packet size\\
		\hline
		Electric protection & 42.56 kb & 8.521 Mbps & 3040b\\ 
		Virtual reality game & 2.16Mb&180 Mbps& 12kb\\
		Video conference & 3.24Mb&162Mbps&12kb\\
		4K video& 7.2Mbps&180Mbps&12kb \\
		\hline
	\end{tabular}
\caption{Characteristics of the 4 classes of flows.}
\label{tab:campus-desc}
\end{table}

Table~\ref{tab:campus} summarizes the delay found for each class and each method. In order to make all methods compute the delays faster and more accurate, we concatenate servers that are crossed by the same sets of flows. This enables the LP method to compute the delays fast (which would not be possible otherwise). 

Because the shaping rate at each server is four times larger than the service rate, TFA++ is outperformed by all the other methods (although slightly by SFA). The delays are divided by approximately 2 between SFA and PLP, and again by 2 between LP and PLP. We also see that the LP method provides a good approximation of the actual worst-case: the last column of Table~\ref{tab:campus} represents the delay obtained when the newtork is simulated and the maximum traffic arrived from time 0. The gap can also be explained by the fact that the DRR service curve can be pessimistic and that this trajectory may not be the one maximizing the delays.

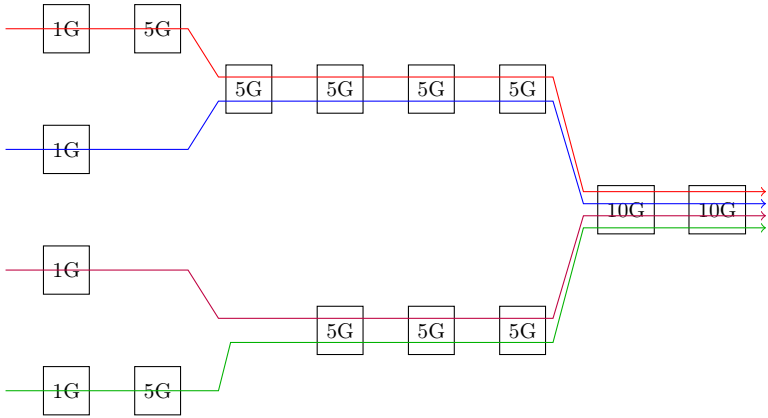
\begin{figure}
	\centering
	\begin{tikzpicture}[scale=0.8, every node/.style={transform shape}, server/.style={shape=rectangle,draw,minimum height=.8cm,inner xsep=1ex}]
	\node[server,name=S1] at (0,3) {1G};
	\node[server,name=S2] at (0,1) {1G};
	\node[server,name=S3] at (0,-1) {1G};
	\node[server,name=S4] at (0,-3) {1G};
	\node[server,name=S5] at (1.5,3) {5G};
	\node[server,name=S6] at (1.5,-3) {5G};
	\node[server,name=S7] at (3,2) {5G};
	\node[server,name=S8] at (4.5,2) {5G};
	\node[server,name=S7] at (6,2) {5G};
	\node[server,name=S8] at (7.5,2) {5G};
	\node[server,name=S8] at (4.5,-2) {5G};
	\node[server,name=S7] at (6,-2) {5G};
	\node[server,name=S8] at (7.5,-2) {5G};
	\node[server,name=S7] at (9.2,0) {10G};
	\node[server,name=S8] at (10.7,0) {10G};
	\draw[red, ->] (-1, 3) -- (2,3) -- (2.5, 2.2) -- (8, 2.2) -- (8.5,0.3) -- (11.5, 0.3); 
	\draw[blue, ->] (-1, 1) -- (2,1) -- (2.5, 1.8) -- (8, 1.8) -- (8.5,0.1) -- (11.5, 0.1); 
	\draw[purple, ->] (-1, -1) -- (2,-1) -- (2.5, -1.8) -- (8, -1.8) -- (8.5,-0.1) -- (11.5, -0.1);
	\draw[green!70!black, ->] (-1, -3) -- (2.5,-3) -- (2.7, -2.2) -- (8, -2.2) -- (8.5,-0.3) -- (11.5, -0.3); 
	\end{tikzpicture}
	\caption{Smart campus network: four classes of traffic arrive according to every path drawn. }
	\label{fig:campus}
\end{figure}

\begin{table}[htbp]
	\centering
	\begin{tabular}{|c||r|r|r|r|r|}
		\hline
		class & TFA++ ($\mu$s) & SFA ($\mu$s)& PLP ($\mu$s)& LP ($\mu$s) & simulation ($\mu$s)\\
		\hline
		EP &	156 &	157 &	129 &	112& 33\\
		VR &	7970 &	6890 &	3700 &	2454 &1911\\
		VC &	11679 &	10260 &	5395 &	3563 &3171\\
		4KV &	26481 &	22881 &	12084 &	7986 &7147\\
		\hline
	\end{tabular}
\caption{Delays with for the four classes of traffic with the different methods, and the simulation of a candidate trajectory for the worst-case delay.}
\label{tab:campus}
\end{table}

\section{Conclusion}
\label{sec:conclusion}
In this paper, we have proposed a new linear program technique for the analysis of FIFO networks, that offers a good trade-off between accuracy of the bounds and tractability. This algorithm does not lead to performance bounds as accurate as the previous ones, but it can be performed in polynomial time, which enables to use it in larger networks from real cases. This new algorithm also improves the performances bounds compared to the other methods from the literature. 

We also presented a linear programming solution to deal with cyclic networks. Although presented for FIFO networks, this solution is valid for the other LP methods used in network calculus. This method improves both the delay bounds and the stability region. 

Comparison with other scalable methods (TFA++, SFA) also enables to have a more precise knowledge of when these bounds are accurate. While SFA is never accurate, we could exhibit cases where TFA++ can provide accurate performance guarantees. It is when the load of the network is small or medium and the maxumum and minimum service rates coincide. 

Through the example of TFA, we saw that the use of the shaping was very important to reduce the performance bounds computed. One research direction would also be to see if SFA can be adapted to take into account the shaping effect into a SFA++ method. Some work has already been done in this direction~\cite{Boy2010}, and the improvement computed when only the flow of interest is shaped. It would be interesting to see if the shaping of the cross-traffic can improve the bound for SFA. 

Concerning the accuracy/tractability trade-off, many questions remain open to move from tractability to scalability. First, the PLP algorithm is tractable, but it might not yet be usable for large network. One step to scalability could be to decompose the network into smaller sub-networks, and recombine these sub-networks for performance computations. Many issues would then have to be solved: what is a good decomposition? In particular, what can be the size of the sub-networks: medium size with PLP, small size with LP?

\bibliographystyle{elsarticle-num}

\end{document}